\documentclass[14pt,a4paper]{article}
\usepackage[authoryear,longnamesfirst]{natbib}
\usepackage{tikz}
\usepackage{amsmath, amsfonts, amssymb, amsthm, mathtools, bbm, bm}
\usepackage{physics}
\usepackage{graphicx}
\usepackage{caption}
\usepackage{subcaption}
\usepackage{microtype}
\usepackage{verbatim}
\usepackage{enumitem}
\usepackage{booktabs}
\usepackage{algorithm}
\usepackage{algorithmic}
\usepackage{etoolbox}
\usepackage[dvipsnames]{xcolor}
\usepackage[colorlinks=true,
            linkcolor=violet,
            citecolor=violet,
            urlcolor=violet]{hyperref}
\usepackage[capitalize,nameinlink]{cleveref}
\usepackage[margin=1in]{geometry}

\newcommand{\bj}{\mathbf{j}}

\newcommand{\be}{\begin{equation}}
\newcommand{\ee}{\end{equation}}
\newcommand{\bq}{\begin{eqnarray}}
\newcommand{\eq}{\end{eqnarray}}

\newtheorem{theorem}{Theorem}
\newtheorem{lemma}{Lemma}

    {\hspace*{\fill}$\Box$\vspace{1.5ex}\par}

\newcommand{\shorttitle}[1]{}
\newcommand{\shortauthors}[1]{}

\makeatletter
\renewcommand{\date}[1]{\gdef\@date{}}
\makeatother

\let\origAuthor\author

\newcommand{\arxivAuthorBlock}{}
\newcounter{arxivAuthorNum}

\renewcommand{\author}[1]{%
  \stepcounter{arxivAuthorNum}%
  \ifnum\value{arxivAuthorNum}>1
    \appto\arxivAuthorBlock{\,,\quad #1}%
  \else
    \gdef\arxivAuthorBlock{#1}%
  \fi
}

\newcommand{\ead}[1]{%
  \appto\arxivAuthorBlock{\,\thanks{\texttt{#1}}}%
}

\newcommand{\affiliation}[1]{%
  \appto\arxivAuthorBlock{\\[0.5em]\small\textit{#1}}%
  \origAuthor{\arxivAuthorBlock}%
  \maketitle%
}

\newenvironment{keywords}%
  {\par\medskip\noindent\textbf{Keywords:~}}%
  {\par\medskip}
\newcommand{\sep}{;\;}

\begin{document}
% ================================================================
%  content.tex
%  Shared body for main.tex (cas-sc journal) and arxiv.tex (article).
%  Do not add \begin{document} or \end{document} here.

%  Front matter
% ----------------------------------------------------------------
\title{MLMC-qDRIFT: Multilevel Variance Reduction for Randomized Quantum Hamiltonian Simulation}
\shorttitle{Multilevel Hamiltonian Simulations}

\author{Pegah Mohammadipour}\ead{pegahmp@psu.edu}
\author{Xiantao Li}\ead{xiantao.li@psu.edu}
\affiliation{Department of Mathematics, The Pennsylvania State University,
             University Park, Pennsylvania 16802, USA}

\shortauthors{P. Mohammadipour, X. Li}

\begin{abstract}
Simulating quantum dynamics is one of the central applications of quantum computing.  For Hamiltonians written as a sum of many terms, deterministic
Trotter--Suzuki product formulas can require applying a large number of term-wise evolutions at each time step, leading to high circuit costs for large or dense systems.  Randomized methods such as qDRIFT offer an alternative: each step samples only one Hamiltonian term, giving a circuit depth with no explicit dependence on the number of terms.  However, when
qDRIFT is used for observable estimation, high precision requires many independent random circuit realizations, resulting in a total gate complexity that scales as $\mathcal{O}(\varepsilon^{-3})$.

We introduce a multilevel Monte Carlo framework for qDRIFT that reduces this sampling overhead.  The method constructs a hierarchy of qDRIFT
estimators with increasing circuit depths and couples adjacent levels by
sharing their random Hamiltonian-term samples.  This coupling makes the
variance of the level differences decay with depth, allowing most samples
to be taken on cheaper, coarse circuits and only a few on expensive, fine
circuits.  We prove that the resulting MLMC-qDRIFT estimator reduces the
total gate complexity for fixed-precision observable estimation from the
standard qDRIFT scaling $\mathcal{O}(\varepsilon^{-3})$ to
$\mathcal{O}(\varepsilon^{-2}\log^2(1/\varepsilon))$, while preserving
qDRIFT's lack of explicit dependence on the number of Hamiltonian terms.
Numerical experiments for spin-chain dynamics confirm the predicted
variance decay and demonstrate the practical gate-count savings of the
multilevel construction.
\end{abstract}


\maketitle

\section{Introduction}
One of the most compelling promises of quantum computing is its ability to
simulate the dynamics of other quantum systems efficiently. This capability,
known as quantum simulation, was among the original motivations for building
quantum computers~\cite{feynman1982simulating}, and it remains one of the
field's most active research directions, with applications spanning
high-energy physics, condensed matter, and quantum chemistry~\cite{IM2014,
aspuru2005simulated, reiher2017elucidating, mcardle2020}.

The fundamental difficulty is associated with the system dimension. For a quantum system of $m$
particles, the Hilbert space grows exponentially in $m$, making direct
numerical integration of the Schrödinger equation,
\begin{equation}
    i\,\frac{d}{dt}\ket{\psi(t)} = H\ket{\psi(t)},
\end{equation}
computationally intractable for all but the smallest systems. Quantum
simulation sidesteps this curse of dimensionality by approximating the
unitary time-evolution operator $U(t) = e^{-iHt}$ as a sequence of
programmable quantum gates, each acting on a small number of qubits. The
central challenge then becomes one of efficiency: how many gates are needed
to approximate $U(t)$ to a target precision $\varepsilon$?

For Hamiltonians of the form $H = \sum_{j=1}^{M} h_j H_j$, which arise
naturally in fermionic systems, lattice models, and sparse scientific
computing problems, a standard gate-based approach is the deterministic
Lie--Trotter--Suzuki product formula~\cite{lloyd1996universal,
suzuki1976generalized, suzuki1991}. These methods enjoy well-understood
error bounds, but each time step involves all the Hamiltonian terms, and therefore the cost can be unfavorable when $M$ is large, especially when commutator structures do not compensate.
This can be prohibitively challenging for dense Hamiltonians or long-range
interactions~\cite{childs2021}. More sophisticated post-Trotter techniques,
including linear combination of unitaries
(LCU)~\cite{childs2012, Chakraborty2024implementingany}, quantum signal
processing~\cite{low2017, motlagh2024}, and
qubitization~\cite{Low2019hamiltonian}, improve the precision scaling but require complex control logic gates, making them challenging to implement on near-term hardware with limited coherence times.

To mitigate these issues, randomised algorithms such as
\textit{quantum stochastic drift protocol (qDRIFT)}~\cite{campbell2019random} have been introduced. The
qDRIFT protocol replaces the deterministic application of all $M$ terms
with stochastic sampling: at each step, a single term $H_j$ is selected
at random with probability proportional to its weight $|h_j|$, and only
that term's gate is applied. The key insight is that the resulting random
channel, when averaged over many realizations, closely approximates the
ideal evolution, with a total gate count that is not directly \emph{dependent of $M$}: Rather, the complexity scales quadratically with $\sum_j \abs{h_j}$, providing great advantage for Hamiltonians with long-range couplings.  
%This makes qDRIFT particularly attractive for Hamiltonians with many terms,
%such as those arising in quantum chemistry or many-body physics.

The main current bottleneck, is a total gate scaling of $\mathcal{O}(\varepsilon^{-3})$: the
circuit depth must grow as $\mathcal{O}(\varepsilon^{-1})$ to suppress the
algorithmic bias due to the random sampling of $H_j$, while $\mathcal{O}(\varepsilon^{-2})$ independent realizations are needed to suppress the statistical variance, i.e., the standard Monte Carlo error. It is precisely the purpose of this paper to reduce the dependence of the gate counts on the precision $\varepsilon.$

\subsection{Standard qDRIFT} \label{subsec:standard_qdrift}

qDRIFT can be viewed as a Monte Carlo discretization of Hamiltonian time evolution. Each random circuit is one sample path. Its average has a deterministic discretization bias, and finite sampling introduces statistical error. This separation is analogous to the bias-variance decomposition in classical Monte Carlo methods for SDEs \cite{jin2025partially}.

We begin with a self-contained review of the qDRIFT protocol, as its
structure is the foundation for our multilevel construction.

Consider a Hamiltonian $H = \sum_{j=1}^{M} h_j H_j$, where $H_j$ are
Hermitian operators with $\|H_j\| = 1$ and $h_j > 0$, after absorbing signs into $H_j$. Define the
one-norm $\lambda = \sum_{j=1}^{M} h_j$ and the sampling probabilities
$p_j = h_j/\lambda$. To approximate the evolution $U(t) = e^{-iHt}$
with $N$ random gates, qDRIFT proceeds as follows:
\begin{enumerate}
    \item \textit{Sample a sequence:} draw $N$ indices
          $\mathbf{j} = (j_1,\dots,j_N)$ i.i.d.\ from $\{p_j\}_{j=1}^M$.
    \item \textit{Construct the circuit:} apply the gates
          $V_{j_k} = e^{-i\tau H_{j_k}}$ in sequence, where
          $\tau = \lambda t/N$ is the step size.
    \item \textit{Estimate the observable:} compute
          $\hat{\mu}^{(\mathbf{j})} = \operatorname{Tr}(O\,V_{j_N}\cdots
          V_{j_1}\,\rho\,V_{j_1}^\dagger\cdots V_{j_N}^\dagger)$
          and average over $n$ independent realizations of $\mathbf{j}$.
\end{enumerate}

Because qDRIFT is randomised, a single realization produces a random
observable estimate whose value depends on the particular gate sequence
drawn. To analyse the algorithm's accuracy, we work with the
\emph{averaged quantum channel}: the map on density matrices obtained
by averaging the random unitary conjugation over all possible gate
sequences. For a single step, this channel is
\begin{equation}
    \mathcal{E}(\rho)
    = \sum_{j=1}^{M} p_j\,e^{-i\tau H_j}\rho\,e^{i\tau H_j},
    \label{eq:qdrift_channel}
\end{equation}
and after $N$ steps the evolution is described by the $N$-fold
composition $\mathcal{E}^N$. The averaged channel plays a dual role.
First, it defines the \emph{algorithmic bias}: the quantity that $n$
independent qDRIFT realizations converge to as $n \to \infty$ is not
the exact expectation value $\mu = \operatorname{Tr}(O\,e^{-iHt}\rho\,
e^{iHt})$, but rather the channel output
$\bar{\mu} = \operatorname{Tr}(O\,\mathcal{E}^N[\rho])$. 
Assuming $\|O\| \leq 1$, the diamond-norm bound~\cite{campbell2019random, chen2021concentration} gives
\begin{equation}
    |\mu - \bar{\mu}|
    \;\le\; \|\mathcal{E}^N - \mathcal{U}(t)\|_\diamond
    \;\le\; \frac{2\lambda^2 t^2}{N},
    \label{eq:qdrift_bias}
\end{equation}
where $\mathcal{U}(t)(\cdot) = e^{-iHt}(\cdot)e^{iHt}$ denotes the exact, 
ideal unitary evolution channel.

To allocate half the total error budget $\varepsilon$ to this bias term, 
we set $|\mu - \bar{\mu}| \leq \varepsilon/2$, which requires 
$N = \mathcal{O}(\lambda^2 t^2/\varepsilon)$ total steps.
Second, the averaged channel
provides the \emph{exact} expected observable value over all random
circuits: by linearity of the trace,
$\mathbb{E}[\hat{\mu}^{(\mathbf{j})}] = \operatorname{Tr}(O\,\mathcal{E}^N[\rho])
= \bar{\mu}$, so $\bar{\mu}$ is both the bias target and the quantity
being estimated by Monte Carlo.

The \emph{statistical variance} arises because individual realizations
fluctuate around $\bar{\mu}$. Let $\sigma^2 = \operatorname{Var}
(\hat{\mu}^{(j)})$ denote the variance of a single trajectory estimate,
combining both the randomness of the sampled gate sequence and the
quantum measurement outcome. Averaging $n$ independent realizations of $\mathbf{j}$
gives an estimator $\hat{\mu} = \frac{1}{n}\sum_{\mathbf{j}} \hat{\mu}^{(\mathbf{j})}$
with variance $\sigma^2/n$. Suppressing this to $\varepsilon^2/2$
requires $n = \mathcal{O}(\sigma^2/\varepsilon^2)$ samples.

The total gate count is therefore
\begin{equation}
    C_{\mathrm{std}}
    = N \times n
    = \mathcal O\!\left(\frac{\lambda^2 t^2 \sigma^2}{\varepsilon^3}\right),
    \label{eq:qdrift_cost}
\end{equation}
which is $\mathcal{O}(\varepsilon^{-3})$ since $\sigma^2 = \mathcal{O}(1)$ for bounded
observables. This scaling is unavoidable within the standard single-level
Monte Carlo framework: the $\varepsilon^{-1}$ factor from the bias and
the $\varepsilon^{-2}$ factor from the variance are controlled by
independent parameters $N$ and $n$, and reducing one does not help the
other. Reducing the total cost requires breaking this independence, which
is precisely what multilevel methods achieve.

\subsection{Main Results}

The central contribution of this work is the \textit{Multilevel Monte Carlo-qDRIFT (MLMC-qDRIFT)} algorithm and
its complexity analysis. To set notation for the multilevel construction,
we write $P_\ell$ for the qDRIFT observable estimator at refinement
level $\ell$, corresponding to a circuit of depth $N_\ell = N_0 \cdot
2^\ell$ for some base depth $N_0$.

We implement the algorithm on
different levels, where each level $\ell$ is realised by an independently
sampled quantum circuit of the corresponding depth. Our main results are as follows.

\textbf{Variance reduction via multilevel coupling.}
We construct a hierarchy of qDRIFT estimators $\{P_\ell\}_{\ell=0}^L$
at increasing levels of accuracy, and couple adjacent levels by sharing
their random gate sequences. We prove that this coupling causes the
variance of each level correction $P_\ell - P_{\ell-1}$ to decay
geometrically as $\mathcal{O}(2^{-\ell})$, a property we call the \emph{index-sharing
variance bound}, while the cost per sample grows at the complementary rate
$\mathcal{O}(2^\ell)$, corresponding to $N_\ell = N_0 \cdot 2^\ell$ gates per circuit.
This balance between variance decay and cost growth is
what enables the MLMC complexity improvement.

\textbf{Gate complexity.}
Applying the MLMC Complexity Theorem of~\cite{giles2015multilevel}
to the qDRIFT setting, we prove that MLMC-qDRIFT achieves a total gate
complexity of
\begin{equation}
    C_{\mathrm{MLMC}} = O\!\left(
      \frac{\lambda^2 t^2}{\varepsilon^2}
      \log^2 (1/\varepsilon)
    \right).
\end{equation}
This reduces the ${\varepsilon^{-3}}$ in \eqref{eq:qdrift_cost} scaling to ${\varepsilon^{-2}}$ up to logarithmic factors. 
 In addition, this complexity is still
independent of the number of Hamiltonian terms $M$, preserving the
key advantage of qDRIFT.

\textbf{Numerical validation.}
We validate the algorithm on a six-qubit Heisenberg XYZ spin chain,
confirming the predicted variance decay rates $\hat{\beta} \approx 0.92$
and the gate complexity crossover at target precision
$\varepsilon \approx 0.02$, beyond which MLMC-qDRIFT delivers
gate-count reductions of up to $28\times$ over standard qDRIFT at
$\varepsilon = 10^{-4}$.

\subsection{Related Work}

Since the introduction of qDRIFT~\cite{campbell2019random}, randomized
methods for quantum simulation have developed along several complementary
directions.  One line of work has focused on understanding the accuracy
of random product formulas beyond their averaged-channel behavior.  In
particular, concentration results show that typical random realizations
can approximate the ideal evolution with high probability, with improved
guarantees in state-dependent settings~\cite{chen2021concentration}.
More recent work has sharpened qDRIFT error bounds for both closed and
open quantum systems~\cite{david2025tighter}.  Randomization has also
been extended beyond time-independent Hamiltonian dynamics, including
continuous qDRIFT for time-dependent Hamiltonians~\cite{berry2020time}
and qDRIFT-type randomized schemes for Lindblad dynamics and thermal
state preparation~\cite{chen2025randomized}.

A second direction concerns the design of more efficient randomized
simulation protocols.  Randomized product formulas improve the analysis
of product-formula simulation by randomizing the ordering of Hamiltonian
terms~\cite{childs2019fasterquantum}, while randomized multiproduct
formulas combine random sampling with higher-order cancellation
mechanisms~\cite{faehrmann2022randomizingmulti}.  Within the qDRIFT
family, high-order randomized compilers such as qSWIFT improve the
precision scaling by using ancillary systems and classical
post-processing~\cite{nakaji2024high}.  Other approaches optimize the
sampling distribution itself: importance sampling generalizes qDRIFT by
allowing arbitrary sampling probabilities while controlling both bias and
statistical fluctuations~\cite{kiss2023importance}, and stochastic
Hamiltonian sparsification randomly removes weak Hamiltonian terms to
reduce implementation cost for electronic-structure Hamiltonians
~\cite{ouyang2020compilation}.  More recently, Fan, Wu, and
Zhang~\cite{fan2025adaptive} proposed an adaptive random compiler in
which the sampling weights are updated using low-order moment information.
These sampling-optimization methods are orthogonal to MLMC-qDRIFT: they
modify the single-level randomized estimator, whereas the multilevel
construction reduces sampling cost by coupling estimators across circuit
depths.

Hybrid deterministic--randomized decompositions provide another route to
reducing simulation cost.  In these methods, large or structurally
important Hamiltonian terms are treated deterministically, typically by
Trotter--Suzuki formulas, while the remaining numerous small terms are
sampled randomly.  This idea appears in partially random Trotter
algorithms~\cite{jin2025partially}, composite simulation
channels~\cite{hagan2023composite}, and composite qDRIFT-product
formulas for both real- and imaginary-time evolution~\cite{pocrnic2024}.
Such methods exploit the same basic observation as qDRIFT: when a
Hamiltonian contains many terms of unequal importance, random sampling can
avoid applying every term at every time step.

Random sampling has also become relevant in algorithms designed for
near-term, early fault-tolerant, and fault-tolerant quantum computation.
While fully coherent Hamiltonian simulation methods based on linear
combinations of unitaries, truncated Taylor series, quantum signal
processing, and qubitization achieve excellent asymptotic precision
scaling~\cite{childs2012,berry2015taylor,low2017,Low2019hamiltonian},
they typically require block encodings, ancilla registers, and coherent
control.  By contrast, randomized approaches trade coherent circuit depth
for classical sampling and statistical estimation.  This tradeoff has
been explored in randomized and statistical variants of phase estimation,
where random sampling can reduce dependence on the number of Hamiltonian
terms or allow algorithmic errors to be suppressed by collecting more
samples rather than increasing circuit depth
~\cite{kivlichan2019phase,wan2022randomized,gunther2025phase}.
qDRIFT-style randomized compilation has also begun to appear inside
larger quantum-simulation workflows.  For example, Piccinelli
et al.~\cite{piccinelli2025sqdrift} combine qDRIFT with sample-based
Krylov quantum diagonalization for quantum chemistry.  Such applications
motivate variance-reduction methods that reduce the repeated-sampling
overhead of randomized time-evolution subroutines.  These developments
reinforce the broader role of randomization as a resource-saving
mechanism in quantum simulation.

Finally, some recent near-term protocols seek to remove discretization
bias altogether for observable estimation.  For example, the method of
Granet and Dreyer constructs an unbiased estimator with bounded average
circuit depth, at the price of an attenuation factor that increases the
number of measurement shots required~\cite{granet2024hamiltonian}.  This
is conceptually related to the present work in that both approaches trade
circuit depth against statistical sampling.  However, our focus is
different: we retain the qDRIFT discretization hierarchy and reduce the
Monte Carlo cost by coupling estimators across levels through a multilevel
Monte Carlo construction.

In total gate complexity, the closest prior work is
qFLO~\cite{watson2024randomly}, which also reduces the standard qDRIFT
scaling by applying Richardson extrapolation to qDRIFT estimates at
multiple step sizes.  The mechanism, however, is different: qFLO targets
the bias expansion in the step size, whereas MLMC-qDRIFT targets the
sampling variance of a hierarchy of coupled estimators.  Related
extrapolation-based strategies have also been developed for product-formula
simulation more broadly; for example, Watson and
Watkins~\cite{watson2025trotter} analyze Trotter error mitigation by
Richardson extrapolation and polynomial interpolation for time-evolved
expectation values.  These works further illustrate the distinction
between bias-reduction mechanisms based on extrapolation and the present
MLMC variance-reduction mechanism.  The qFLO complexity bound relies on
sufficient smoothness of the observable as a function of the step size
over the extrapolation range; its analysis establishes this structure
near $\tau=0$, while an end-to-end extrapolation guarantee requires
control on the whole interval used by the extrapolator.  This smoothness
issue is precisely what is made explicit in the Lindblad extension
~\cite{mohammadipour2025reducing}.  By contrast, MLMC-qDRIFT uses only
the first-order qDRIFT bias bound and the variance decay of the level
differences.  The two approaches are therefore complementary rather than
competing: Richardson-type bias cancellation can, in principle, be
combined with multilevel variance reduction, as in classical multilevel
Richardson--Romberg methods~\cite{lemaire2017multilevel}.

We introduce the Multilevel Monte Carlo framework for qDRIFT and prove its computational advantage in Section~\ref{sec:mlmc}, then characterize the regime in which this advantage holds in Section~\ref{sec:complexity}. Section \ref{sec:numerics} provides numerical validation of our theoretical results, and we conclude the paper in Section \ref{sec:discussion}. Further technical details of the MLMC methodology and supplementary numerical analyses are provided in the appendices.

\section{Multilevel Monte Carlo for qDRIFT}
\label{sec:mlmc}

As alluded to in \cref{subsec:standard_qdrift}, standard qDRIFT requires a
total gate count of $\mathcal{O}(\varepsilon^{-3})$ to estimate
$P = \operatorname{Tr}(O\,e^{-iHt}\rho\,e^{iHt})$ to precision
$\varepsilon$, with the bias and variance contributions controlled by
independent parameters $N$ and $n$. The key structural observation
motivating the multilevel approach is that the randomised dynamics of
qDRIFT, each circuit driven by a randomly sampled sequence of
Hamiltonian terms, is analogous to a classical SDE driven by random
noise, for which ~\cite{giles2015multilevel} showed that
correlating estimators across a hierarchy of discretization levels
reduces the total cost to $\mathcal{O}(\varepsilon^{-2}\log^2\varepsilon^{-1})$, which is the primary motivation of the current work.
In this section we develop the quantum analogue of this construction,
derive the corresponding variance bounds, and establish the resulting
gate complexity.

\subsection{The MLMC Estimator}
Rather than a single sequence of random Hamiltonians, 
our MLMC-qDRIFT approach operates by constructing a hierarchy of approximation levels indexed
by an integer $\ell \in \{0, 1, \dots, L\}$, where $\ell = 0$ is the
coarsest and $\ell = L$ the finest, chosen based on the desired precision. The levels are defined according to the qDRIFT step angle
\begin{equation}
    \tau_\ell = \frac{\lambda t}{N_\ell}, \qquad
    N_\ell = N_0 \cdot 2^\ell .
\end{equation}
Here $\tau_\ell$ includes the total Hamiltonian weight. Each sampled gate is
$e^{-i\tau_\ell H_j}$ and thus the physical evolution time remains $t$. In particular,  $N_\ell$
determines the circuit depth at level $\ell$, so that finer levels correspond to
smaller step sizes and higher accuracy. The transition for level $\ell$ to level $\ell+1$ amounts to reducing the step size by half. It is also worth noting that the step size has taken into account the total Hamiltonian weight $\lambda$.

At each level $\ell$, we define the
random variable corresponding to the expectation of an observable $O$,
\begin{equation}
    P_\ell := \operatorname{Tr}\!\bigl(O\,U_\ell\,\rho\,U_\ell^\dagger\bigr),
    \label{eq:P_ell_def}
\end{equation}
where $U_\ell$ is a qDRIFT circuit of depth $N_\ell = N_0 \cdot 2^\ell$
driven by a randomly sampled index sequence, 
\begin{equation}\label{j-seq}
    \mathbf{j}= (j_1, j_2, \cdots, j_{N_\ell}), 
\end{equation}
for the Hamiltonians $H_j$ with probability proportional to $\abs{h_j}$,
so that
$\mathbb{E}[P_\ell] \to \operatorname{Tr}(O\,e^{-iHt}\rho\,e^{iHt})$
as $\ell \to \infty$. Rather than relying on a single high-accuracy
simulation at $\ell = L$, MLMC distributes the computational burden
across this hierarchy, concentrating most samples at low-cost coarse
levels, via the telescoping identity
\begin{equation}
    \mathbb{E}[P_L]
    = \mathbb{E}[P_0]
      + \sum_{\ell=1}^{L} \mathbb{E}[P_\ell - P_{\ell-1}],
    \label{eq:telescoping}
\end{equation}
where each correction $\mathbb{E}[P_\ell - P_{\ell-1}]$, associated with
step size $\tau_\ell $, can be estimated
using \emph{shared randomness} between adjacent levels to dramatically
reduce its variance.
Each level $\ell$ is characterised by a gate count $N_\ell = N_0 \cdot
2^\ell$, where $N_0$ is the base step count and each level doubles the
circuit depth. 
For the  independent realization $\mathbf{j}$ at level $\ell$,
determined by a random index in \eqref{j-seq} with  
i.i.d.\ distribution from $\{p_i\}_{i=1}^M$, the level estimator is
\begin{equation}
    P_\ell^{(\bj,\ell)}
    = \operatorname{Tr}\!\bigl(O\,U_\ell^{(\bj,\ell)}\,\rho\,
      U_\ell^{(\bj,\ell)\dagger}\bigr),
    \quad
    U_\ell^{(\bj,\ell)}
    = \prod_{k=1}^{N_\ell}
      \exp\!\Bigl(-i\,\frac{\lambda t}{N_\ell}\,H_{j_k}\Bigr).
    \label{eq:level_estimator}
\end{equation}

Consequently, the full MLMC estimator is
\begin{equation}
    \widehat{Y}
    = \underbrace{\frac{1}{n_0}\sum_{n=1}^{n_0}
        P_0^{(\bj_n,0)}}_{\text{coarse base}}
    + \sum_{\ell=1}^{L}
      \underbrace{
        \frac{1}{n_\ell}\sum_{n=1}^{n_\ell}
        Y_\ell^{(\bj_n,\ell)}
      }_{\text{level-}\ell\text{ correction}},
    \label{eq:mlmc_estimator}
\end{equation}
where $Y_\ell^{(\bj,\ell)} := P_\ell^{(\bj,\ell)} - P_{\ell-1}^{(\tilde \bj,\ell)}$
and $n_\ell$ is the number of independent trajectories at level $\ell$.
Equivalently, 
\begin{equation}
\widehat{Y} = \frac{1}{n_0}\sum_{n=1}^{n_0} P_0^{(\bj_n,0)}
+ \sum_{\ell=1}^{L} \frac{1}{n_\ell}\sum_{n=1}^{n_\ell} Y_\ell^{(\bj_n,\ell)}.
\end{equation}

Our notion $(\bj, \ell)$ emphasizes that the sequence $\bj$ has length depending on the level $\ell$. More importantly, $(\tilde \bj,\ell)$ corresponds to a sub-sampling of $\bj$ at level $\ell$, which will be elaborated upon next.  For each sample $\bj$ at level $\ell \geq 1$, the coupled pair
$(P_\ell^{(\bj,\ell)}, P_{\ell-1}^{(\bj,\ell)})$ is computed using the
\emph{same} underlying index sequence $(j_1,\dots,j_{N_\ell})$ to
maximize their correlation and so create a reduction in
$\operatorname{Var}(P_\ell - P_{\ell-1})$. The estimator $\widehat{Y}$
is unbiased for $\mathbb{E}[P_L]$, which converges to
$\operatorname{Tr}(O\,e^{-iHt}\rho\,e^{iHt})$ as $N_L \to \infty$.

In contrast, if the fine and coarse estimators at level $\ell$ were sampled
independently, $\operatorname{Var}(P_\ell - P_{\ell-1})
= \operatorname{Var}(P_\ell) + \operatorname{Var}(P_{\ell-1})$,
offering no advantage over standard Monte Carlo. The variance reduction
comes entirely from the coupling: by driving both estimators with the
same random index sequence, the correction variance
$V_\ell = \operatorname{Var}(P_\ell - P_{\ell-1})$ often decays geometrically
as $N_\ell$ grows, allowing fewer samples at each
successive level. The explicit bound on $V_\ell$ is established in
Lemma~\ref{lemma:index_sharing} below for the qDRIFT, after the coupling is defined;
the optimal allocation formula for $n_\ell$ is then derived directly
from that bound.

\begin{center}
\hrule
\vspace{1mm}
\textbf{Algorithm: Multilevel Monte Carlo qDRIFT}
\vspace{1mm}
\hrule
\begin{algorithmic}[1]
\STATE \textbf{Input:} Hamiltonian $\{H_j, h_j\}_{j=1}^M$, time $t$,
       error $\varepsilon$, levels $L$, base gate count $N_0$
\STATE \textbf{Initialise:} $\lambda = \sum_{j=1}^{M} |h_j|$,\;
       $p_j = |h_j|/\lambda$,\;
       $N_\ell = N_0 \cdot 2^\ell$,\;
       $\tau_\ell = \lambda t / N_\ell$ \quad for $\ell = 0,\dots,L$
\FOR{level $\ell = 0$ \TO $L$}
    \STATE Determine optimal sample count $n_\ell$ for level $\ell$
           via~\eqref{eq:optimal_allocation}
    \FOR{$n = 1$ \TO $n_\ell$}
        \STATE Sample index sequence
               $\mathbf{j}_n = (j_{n,1}, j_{n,2}, \dots, j_{n,N_\ell})$
               i.i.d.\ from $\{p_j\}_{j=1}^M$
        \STATE Compute fine unitary:
               $U_\ell^{(\mathbf{j}_n,\ell)} \gets
               \prod_{k=1}^{N_\ell}\exp\!\bigl(-i\,\tau_\ell\,H_{j_{n,k}}\bigr)$
        \STATE $P_\ell^{(\mathbf{j}_n,\ell)} \gets
               \operatorname{Tr}\!\bigl(O\,U_\ell^{(\mathbf{j}_n,\ell)}\,
               \rho\,U_\ell^{(\mathbf{j}_n,\ell)\dagger}\bigr)$
        \IF{$\ell > 0$}
            \STATE Obtain $\tilde{\mathbf{j}}_n =
                   (j_{n,1},\, j_{n,3},\, \dots,\, j_{n,N_\ell - 1})$
                   by sub-sampling every other index from $\mathbf{j}_n$
                   (index-sharing; see \cref{subsec:index_sharing})
            \STATE Compute coarse unitary:
                   $U_{\ell-1}^{(\tilde{\mathbf{j}}_n,\ell)} \gets
                   \prod_{k=1}^{N_{\ell-1}}
                   \exp\!\bigl(-i\,\tau_{\ell-1}\,H_{j_{n,2k-1}}\bigr)$
            \STATE $P_{\ell-1}^{(\tilde{\mathbf{j}}_n,\ell)} \gets
                   \operatorname{Tr}\!\bigl(O\,U_{\ell-1}^{(\tilde{\mathbf{j}}_n,\ell)}\,
                   \rho\,U_{\ell-1}^{(\tilde{\mathbf{j}}_n,\ell)\dagger}\bigr)$
            \STATE $Y_\ell^{(n)} \gets
                   P_\ell^{(\mathbf{j}_n,\ell)} - P_{\ell-1}^{(\tilde{\mathbf{j}}_n,\ell)}$
        \ELSE
            \STATE $Y_0^{(n)} \gets P_0^{(\mathbf{j}_n,0)}$
        \ENDIF
    \ENDFOR
    \STATE $\widehat{\Delta E}_\ell \gets
           \dfrac{1}{n_\ell}\sum_{n=1}^{n_\ell} Y_\ell^{(n)}$
\ENDFOR
\STATE \textbf{Output:}
       $\widehat{Y} = \sum_{\ell=0}^{L} \widehat{\Delta E}_\ell
       \approx \operatorname{Tr}(O\,e^{-iHt}\rho\,e^{iHt})$
\end{algorithmic}
\hrule
\end{center}

\subsection{Index-Sharing Coupling}
\label{subsec:index_sharing}

We now define the coupling precisely. Let $\tau_0$ be the coarsest step size; 
the step-size hierarchy is then
\begin{equation}
    \tau_\ell := \tau_0\,2^{-\ell},
    \qquad \ell = 0, 1, \dots, L,
    \label{eq:stepsize}
\end{equation}
with $N_\ell = \lambda t/\tau_\ell = N_0\,2^\ell$, where $N_0 = \lambda t/\tau_0$ 
is the corresponding coarsest gate count. Each
increase in level doubles the circuit depth and halves the time-step duration.
The framework generalises to any integer doubling factor.

It suffices to explain the coupling scheme for two successive levels $\ell-1$ and $\ell$.
Because the coarse step duration $\tau_{\ell-1} = 2\tau_\ell$ is twice
the fine step duration, each coarse step corresponds to exactly two fine
steps. For the $k$-th coarse step ($k = 1,\dots,N_{\ell-1}$), two
i.i.d.\ indices $j_{n,2k-1}$ and $j_{n,2k}$ are drawn from
$\{p_j\}_{j=1}^M$. The fine and coarse levels process these shared
indices by utilizing the following simple strategy
\begin{itemize}
    \item \textit{Fine level:} applies both indices as two sequential
          gates of duration $\tau_\ell$:
          \begin{equation}
              U_\ell^{(\mathbf{j}_n,\ell)}
              = \prod_{k=1}^{N_{\ell-1}}
                e^{-i\tau_\ell H_{j_{n,2k}}}\,
                e^{-i\tau_\ell H_{j_{n,2k-1}}}.
              \label{eq:fine_unitary}
          \end{equation}
    \item \textit{Coarse level:} discards the even index and applies
          only the odd index for the full coarse duration
          $\tau_{\ell-1} = 2\tau_\ell$:
          \begin{equation}
              U_{\ell-1}^{(\tilde{\mathbf{j}}_n,\ell)}
              = \prod_{k=1}^{N_{\ell-1}}
                e^{-i\tau_{\ell-1} H_{j_{n,2k-1}}}.
              \label{eq:coarse_unitary}
          \end{equation}
\end{itemize}
The shared odd indices $\tilde{\mathbf{j}}_n = (j_{n,1}, j_{n,3}, \dots,
j_{n,N_\ell-1})$ create a strong correlation between the two levels
while preserving the correct qDRIFT marginal distribution at the coarse
level. Although the two block evolutions do not agree pathwise, their
leading difference is a sum of mean-zero local fluctuations, whose
accumulated mean-square size decreases with the step size. This is the
mechanism behind the variance bound below.

\begin{figure}[htbp]
    \centering
    \begin{tikzpicture}[
        element/.style={circle, minimum size=8mm, inner sep=1pt,
          font=\large},
        level label/.style={font=\bfseries, anchor=east},
        arrow/.style={->, thick, shorten >=2pt, shorten <=2pt, >=stealth}]
        \node[level label] at (0, 2)
          {Level $\ell$ ($N_\ell = N_0\,2^\ell$): $\mathbf{j}_n$ };
        \node[element] (j1)    at (1.5,  2) {$j_{n,1}$};
        \node[element] (j2)    at (3.0,  2) {$j_{n,2}$};
        \node[element] (j3)    at (4.5,  2) {$j_{n,3}$};
        \node[element] (j4)    at (6.0,  2) {$j_{n,4}$};
        \node[element] (dots1) at (7.5,  2) {$\cdots$};
        \node[element] (jNm1)  at (9.0,  2) {$j_{n,N-1}$};
        \node[element] (jN)    at (10.5, 2) {$j_{n,N}$};
        \node[level label] at (0, 0)
          {Level $\ell-1$ ($N_{\ell-1} = N_0\,2^{\ell-1}$):
          $\tilde{\mathbf{j}}_n$};
        \node[element] (j1s)   at (1.5, 0) {$j_{n,1}$};
        \node[element] (j3s)   at (4.5, 0) {$j_{n,3}$};
        \node[element] (dots2) at (7.5, 0) {$\cdots$};
        \node[element] (jNm1s) at (9.0, 0) {$j_{n,N-1}$};
        \draw[arrow] (j1)   -- (j1s);
        \draw[arrow] (j3)   -- (j3s);
        \draw[arrow] (jNm1) -- (jNm1s);
    \end{tikzpicture}
    \caption{Index-sharing coupling. The full index sequence
      $\mathbf{j}_n = (j_{n,1}, j_{n,2}, \dots, j_{n,N_\ell})$ is used
      at fine level~$\ell$, while only the odd-indexed elements
      $\tilde{\mathbf{j}}_n = (j_{n,1}, j_{n,3}, \dots, j_{n,N_\ell-1})$
      are retained at coarser level~$\ell-1$, each applied for twice the
      step duration so that both levels simulate the same total time $t$.}
    \label{fig:sequence_hierarchy}
\end{figure}
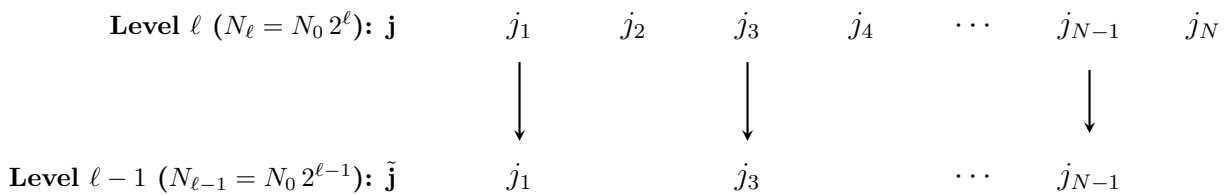

The efficiency of MLMC is governed by two competing rates: the decay of
the correction variance with the level $\ell$, and the growth of the
circuit cost with $\ell$. We first establish the variance decay for
the index-sharing coupling.

\begin{lemma}[Index-Sharing Variance Bound]
\label{lemma:index_sharing}
(Index-Sharing Variance Bound)
Let $P_\ell^{(\mathbf{j}_n,\ell)}$ and $P_{\ell-1}^{(\tilde{\mathbf{j}}_n,\ell)}$
be the random variables from the index-sharing coupling. Under
$\|H_j\| \le 1$ and $\|O\| \le 1$, the correction variance satisfies
\begin{equation}
    V_\ell
    =
    \operatorname{Var}\!\left(
        P_\ell^{(\mathbf{j}_n,\ell)}
        - P_{\ell-1}^{(\tilde{\mathbf{j}}_n,\ell)}
    \right)
    \le
    C_{\mathrm{IS}}\,\frac{\lambda^2t^2}{N_0}\,2^{-\ell},
    \label{eq:variance_bound}
\end{equation}
where $C_{\mathrm{IS}}$ is a universal constant independent of $\ell$,
$N_0$, and the number of Hamiltonian terms $M$.

\end{lemma}

\begin{proof}[Proof sketch]
The proof proceeds by first bounding the correction variance in terms of
the mean-square difference between the coupled fine and coarse unitaries.
Duhamel's principle is then used to express the local unitary difference
over each coarse block. Under the index-sharing coupling, the leading
local fluctuations have zero mean and accumulate in mean square over the
$N_{\ell-1}$ coarse blocks, yielding a total bound proportional to
$N_{\ell}^{-1} = N_0^{-1} 2^{-\ell}$. The complete proof is provided in
\cref{app:proof_of_lemma}.
\end{proof}

With the variance bound established, we now determine the optimal sample
allocation. For notational uniformity, set $Y_0^{(n)} = P_0^{(\mathbf{j}_n,0)}$
and $V_0 = \operatorname{Var}(Y_0^{(n)})$. Let
\begin{equation}
    C_0 := N_0,
    \qquad
    C_\ell := N_\ell + N_{\ell-1},
    \quad \ell \ge 1,
    \label{eq:level_costs}
\end{equation}
denote the gate cost of one sample at each level. The counts $n_\ell$
are chosen to minimize the total expected gate cost
\[
    \sum_{\ell=0}^L n_\ell C_\ell
\]
subject to the variance budget
\[
    \sum_{\ell=0}^L \frac{V_\ell}{n_\ell}
    \le \frac{\varepsilon^2}{2}.
\]
Applying the Cauchy--Schwarz inequality to this constrained
optimization~\cite{giles2015multilevel} gives
\begin{equation}
    n_\ell
    =
    \left\lceil
        \frac{2}{\varepsilon^2}
        \sqrt{\frac{V_\ell}{C_\ell}}
        \sum_{\ell'=0}^{L}
        \sqrt{V_{\ell'} C_{\ell'}}
    \right\rceil,
    \qquad
    \ell = 0,\dots,L.
    \label{eq:optimal_allocation}
\end{equation}
In practice, the variances $V_\ell$ may be estimated from a small pilot
run before the main computation. In the theoretical analysis below, we
use the analytical bound~\eqref{eq:variance_bound} for $\ell \ge 1$.

\begin{theorem}[Complexity of Index-Sharing MLMC-qDRIFT]
\label{thm:index_sharing_complexity}
Let
\[
    P = \operatorname{Tr}\!\left(O\,e^{-iHt}\rho\,e^{iHt}\right),
\]
and assume $\|H_j\| \le 1$ and $\|O\| \le 1$. Choose the finest level
$L$ so that the qDRIFT bias satisfies
\begin{equation}
    \bigl|\mathbb{E}[P_L] - P\bigr|
    \le \frac{\varepsilon}{\sqrt{2}},
    \label{eq:bias_budget}
\end{equation}
and choose the sample counts $n_\ell$ according
to~\eqref{eq:optimal_allocation}, so that
\begin{equation}
    \sum_{\ell=0}^{L}\frac{V_\ell}{n_\ell}
    \le \frac{\varepsilon^2}{2}.
    \label{eq:variance_budget}
\end{equation}
Then the MLMC estimator $\widehat{Y}$ of~\eqref{eq:mlmc_estimator}
satisfies
\begin{equation}
    \mathbb{E}\!\left[(\widehat{Y} - P)^2\right]^{1/2}
    \le \varepsilon.
    \label{eq:mse_bound}
\end{equation}
Moreover, under the index-sharing variance bound~\eqref{eq:variance_bound},
the expected total gate count obeys
\begin{equation}
    \mathbb{E}[C]
    =
    \mathcal{O}\!\left(
        t^2\lambda^2\,
        \varepsilon^{-2}
        \log^2(1/\varepsilon)
    \right).
    \label{eq:complexity_bound}
\end{equation}
\end{theorem}

\begin{proof}[Proof sketch]
The MLMC estimator is unbiased for $\mathbb{E}[P_L]$, so its
mean-square error decomposes as
\begin{equation}
    \mathbb{E}\!\left[(\widehat{Y} - P)^2\right]
    =
    \bigl(\mathbb{E}[P_L] - P\bigr)^2
    +
    \sum_{\ell=0}^{L}\frac{V_\ell}{n_\ell}.
    \label{eq:mse_decomposition}
\end{equation}
The choices~\eqref{eq:bias_budget} and~\eqref{eq:variance_budget}
therefore imply~\eqref{eq:mse_bound}.

It remains to estimate the cost. For $\ell \ge 1$,
Lemma~\ref{lemma:index_sharing} gives $V_\ell = \mathcal{O}(2^{-\ell})$,
while the cost of one coupled correction sample is
\[
    C_\ell = N_\ell + N_{\ell-1}
    = \frac{3}{2}N_0\,2^\ell
    = \mathcal{O}(2^\ell).
\]
Thus the MLMC variance and cost exponents satisfy $\beta = \gamma = 1$.
Equivalently, each summand in the Giles sum
\begin{equation}
    \mathcal{S}
    :=
    \sum_{\ell=0}^{L}\sqrt{V_\ell C_\ell}
\end{equation}
is bounded uniformly in $\ell$, up to constants depending on $t\lambda$.
Hence $\mathcal{S} = \mathcal{O}(t\lambda L)$, and the optimal
allocation gives
\begin{equation}
    \mathbb{E}[C]
    \le
    \frac{2}{\varepsilon^2}\mathcal{S}^2
    =
    \mathcal{O}\!\left(
        t^2\lambda^2\,
        \varepsilon^{-2} L^2
    \right).
\end{equation}
Finally, the qDRIFT channel-norm bias bound gives
\[
    \bigl|\mathbb{E}[P_L] - P\bigr|
    =
    \mathcal{O}(N_L^{-1}),
\]
so choosing $L = \mathcal{O}(\log(1/\varepsilon))$ enforces
\eqref{eq:bias_budget}. Substituting this choice of $L$ yields
\eqref{eq:complexity_bound}. The complete proof is given in
\cref{app:MLMC_Complexity}.
\end{proof}

\subsection{Quantum Measurement of MLMC-qDRIFT Corrections}
\label{sec:augmented}

The preceding subsection established the MLMC variance reduction at the
level of exact pathwise observables. Namely, for each sampled qDRIFT
index sequence, the coupled fine and coarse circuits define a level
correction
\[
    Y_\ell^{(n)} = P_\ell^{(\mathbf{j}_n,\ell)} - P_{\ell-1}^{(\tilde{\mathbf{j}}_n,\ell)},
\]
and the index-sharing coupling makes the variance of this correction
decay with the level. On an actual quantum device, however, the
observable expectations $P_\ell^{(\mathbf{j}_n,\ell)}$ and
$P_{\ell-1}^{(\tilde{\mathbf{j}}_n,\ell)}$ are not available directly.
They must be estimated from measurement outcomes, which introduces an
additional source of variance.

This measurement issue is specific to the quantum implementation of
MLMC. In classical path sampling, one can evaluate the coupled path
correction directly once the two paths have been generated. In the
quantum setting, by contrast, the path correction is encoded in quantum
states and must be extracted through measurements. A naive measurement
strategy can therefore destroy the variance decay obtained from the
classical MLMC coupling. The purpose of this subsection is to introduce
an augmented-state construction that preserves the MLMC variance
reduction at the measurement level.

Throughout this subsection, we use the convention that the qDRIFT step
size $\tau_\ell$ already includes the factor $\lambda = \sum_j |h_j|$.
Thus a sampled qDRIFT gate at level $\ell$ is written as
\[
    e^{-i\tau_\ell H_j}.
\]
For a fixed level $\ell \ge 1$, write $\tau = \tau_\ell$, so that the
coarse step size is $\tau_{\ell-1} = 2\tau$.

\bigskip

\subsubsection{Coupled fine and coarse paths}

Fix a level $\ell \ge 1$, and let
\[
    K := N_{\ell-1}
\]
be the number of coarse steps. We sample a coupled index sequence
\begin{equation}
    \mathbf{j}_n = (j_{n,1}, j_{n,2}, \dots, j_{n,2K})
    \label{eq:index_seq}
\end{equation}
i.i.d.\ from the qDRIFT distribution $p_j = |h_j|/\lambda$. The fine
path uses all $2K = N_\ell$ sampled indices, while the coarse path uses
only the odd-indexed subsequence $\tilde{\mathbf{j}}_n = (j_{n,1},
j_{n,3}, \dots, j_{n,2K-1})$. For the $m$-th coarse block, define
\[
    a_m := j_{n,2m-1},
    \qquad
    b_m := j_{n,2m},
    \qquad
    m = 1,\dots,K.
\]
The fine and coarse block propagators are
\begin{equation}
    U_{f,m}
    :=
    e^{-i\tau H_{b_m}}e^{-i\tau H_{a_m}},
    \qquad
    U_{c,m}
    :=
    e^{-i(2\tau)H_{a_m}}.
    \label{eq:pair_propagators}
\end{equation}
Thus the fine path applies two sampled gates of step size $\tau$,
whereas the coarse path applies the first sampled Hamiltonian for the
doubled step size $2\tau$.

Starting from the common initial state
\[
    |\psi^{(\ell)}_0\rangle
    =
    |\psi^{(\ell-1)}_0\rangle
    =
    |\psi_0\rangle,
\]
we define the synchronised iterates
\begin{equation}
    |\psi^{(\ell)}_{2m}\rangle
    =
    U_{f,m}|\psi^{(\ell)}_{2m-2}\rangle,
    \qquad
    |\psi^{(\ell-1)}_m\rangle
    =
    U_{c,m}|\psi^{(\ell-1)}_{m-1}\rangle,
    \label{eq:iterates}
\end{equation}
for $m = 1,\dots,K$. After block $m$, the fine path has taken $2m$
steps and the coarse path has taken $m$ steps, corresponding to the
same physical simulation time.

\bigskip

\subsubsection{Why direct measurement loses the MLMC variance decay}

For a fixed sampled sequence $\mathbf{j}_n$, the level correction is
the difference of two quantum expectation values,
\begin{equation}
    Y_\ell^{(n)}
    :=
    \langle\psi^{(\ell)}_{2K}|O|\psi^{(\ell)}_{2K}\rangle
    -
    \langle\psi^{(\ell-1)}_K|O|\psi^{(\ell-1)}_K\rangle.
    \label{eq:Yell_def}
\end{equation}
At the level of exact expectations, the index-sharing coupling reduces
the variance of $Y_\ell^{(n)}$ as $\ell$ increases. On quantum
hardware, however, each expectation value must be inferred from
measurement outcomes, and those outcomes introduce shot noise.

To see the obstruction, suppose first that the fine and coarse circuits
are measured separately. For the simplicity of the illustration, let
$x_f$ and $x_c$ denote single-shot measurement outcomes of $O$ on the
fine and coarse states, respectively. Even when the two circuits are
driven by the same sampled index sequence $\mathbf{j}_n$, the
measurement outcomes are independent once the two circuits are measured
separately. Hence
\[
    \operatorname{Var}(x_f - x_c)
    =
    \operatorname{Var}(x_f) + \operatorname{Var}(x_c).
\]
For a bounded observable $O$, this variance is $\mathcal{O}(1)$, even
when the two expectation values are close and the ideal correction
$Y_\ell^{(n)}$ is small. Thus separate measurement destroys the
variance decay achieved by the MLMC coupling at the level of exact
pathwise expectations.

A natural next idea is to preserve the coupling during measurement by
running the two paths coherently in the same circuit. One can introduce
a control qubit and prepare a superposition in which the $|0\rangle$
branch follows the fine qDRIFT path and the $|1\rangle$ branch follows
the coupled coarse path:
\begin{equation}
    |\Phi^{(\ell)}_K\rangle
    =
    \frac{1}{\sqrt{2}}
    \left(
        |0\rangle_c \otimes |\psi^{(\ell)}_{2K}\rangle
        +
        |1\rangle_c \otimes |\psi^{(\ell-1)}_K\rangle
    \right).
    \label{eq:phi_before_meas}
\end{equation}
This construction correctly encodes the difference as a single
expectation value, since
\begin{equation}
    \langle\Phi^{(\ell)}_K|
    2(Z_c \otimes O)
    |\Phi^{(\ell)}_K\rangle
    =
    Y_\ell^{(n)}.
    \label{eq:ancilla_expectation}
\end{equation}
However, it still does not solve the shot-noise problem. If $O$ is a
Pauli observable, then $O^2 = I$, and therefore
\[
    \bigl[2(Z_c \otimes O)\bigr]^2 = 4I.
\]
Consequently,
\begin{equation}
    \operatorname{Var}\bigl(2(Z_c \otimes O)\bigr)
    =
    4 - Y_\ell^{(n)}{}^2
    =
    \mathcal{O}(1).
    \label{eq:bad_var}
\end{equation}
The mean $Y_\ell^{(n)}$ decreases at fine levels, but the single-shot
variance remains level independent, leading to a small
signal-to-noise issue. The correction is therefore hidden by a fixed
measurement noise floor.

\subsubsection{Augmented difference-state construction}

The key implementation idea of this work is to encode the level
correction directly into an augmented quantum state. Instead of
measuring $O$ on the fine and coarse states separately, or measuring a
fixed-norm difference observable on a controlled superposition, we
evolve the fine--coarse difference state alongside the coarse state. A
level-dependent scaling is then chosen so that the block observable used
to recover the correction has norm decreasing with $\ell$. This makes
the measurement shot noise decay at the same rate as the MLMC correction
variance.

Define the difference state at the $m$-th coarse synchronisation time by
\begin{equation}
    |e^{(\ell)}_m\rangle
    :=
    |\psi^{(\ell)}_{2m}\rangle
    -
    |\psi^{(\ell-1)}_m\rangle,
    \qquad
    |e^{(\ell)}_0\rangle = 0.
    \label{eq:error_state}
\end{equation}
Here we use $m$ to label the time stepping.
Using~\eqref{eq:iterates}, the pair
$(|e^{(\ell)}_m\rangle, |\psi^{(\ell-1)}_m\rangle)$ satisfies the exact
recursion
\begin{align}
    |e^{(\ell)}_m\rangle
    &=
    U_{f,m}|e^{(\ell)}_{m-1}\rangle
    +
    \bigl(U_{f,m} - U_{c,m}\bigr)
    |\psi^{(\ell-1)}_{m-1}\rangle,
    \label{eq:error_recursion}
    \\
    |\psi^{(\ell-1)}_m\rangle
    &=
    U_{c,m}|\psi^{(\ell-1)}_{m-1}\rangle.
    \label{eq:coarse_recursion}
\end{align}
The forcing term
\[
    \bigl(U_{f,m} - U_{c,m}\bigr)|\psi^{(\ell-1)}_{m-1}\rangle
\]
is the local mismatch between the coupled fine and coarse qDRIFT paths.

Let $\zeta_\ell > 0$ be a scaling parameter, which will be specified
later. We define the augmented state
\begin{equation}
    |\chi^{(\ell)}_m\rangle
    :=
    \begin{bmatrix}
        \zeta_\ell |e^{(\ell)}_m\rangle \\[1mm]
        |\psi^{(\ell-1)}_m\rangle
    \end{bmatrix}
    \in \mathbb{C}^{2d},
    \qquad
    |\chi^{(\ell)}_0\rangle
    =
    \begin{bmatrix}
        0 \\ |\psi_0\rangle
    \end{bmatrix},
    \label{eq:augmented_state}
\end{equation}
where $d$ is the dimension of the physical Hilbert space. If the
physical system is represented on $q$ qubits, then $d = 2^q$, and the
augmented state uses one additional block qubit.

Combining~\eqref{eq:error_recursion} and~\eqref{eq:coarse_recursion}
gives the block recursion
\begin{equation}
    |\chi^{(\ell)}_m\rangle
    =
    \mathcal{W}^{(\ell)}_m
    |\chi^{(\ell)}_{m-1}\rangle,
    \qquad
    \mathcal{W}^{(\ell)}_m
    =
    \begin{bmatrix}
        U_{f,m} & \zeta_\ell(U_{f,m} - U_{c,m}) \\
        0       & U_{c,m}
    \end{bmatrix}.
    \label{eq:augmented_recursion}
\end{equation}
The map $\mathcal{W}^{(\ell)}_m$ is generally not unitary because of
the off-diagonal block. Its dilation or block-encoding implementation
will be discussed separately; the present subsection focuses only on the
variance mechanism enabled by the augmented representation.

\bigskip

\subsubsection{Scaled block observable}

The augmented state recovers the MLMC correction as a single quadratic
form. Since
\[
    |\psi^{(\ell)}_{2K}\rangle
    =
    |\psi^{(\ell-1)}_K\rangle + |e^{(\ell)}_K\rangle,
\]
the correction~\eqref{eq:Yell_def} can be written as
\begin{equation}
    Y_\ell^{(n)}
    =
    \langle e^{(\ell)}_K|O|e^{(\ell)}_K\rangle
    +
    \langle e^{(\ell)}_K|O|\psi^{(\ell-1)}_K\rangle
    +
    \langle\psi^{(\ell-1)}_K|O|e^{(\ell)}_K\rangle.
    \label{eq:Yell_expanded}
\end{equation}
To express it as a single expectation, we define the Hermitian block
observable
\begin{equation}
    \widehat{O}_\ell
    :=
    \begin{bmatrix}
        \zeta_\ell^{-2}O & \zeta_\ell^{-1}O \\
        \zeta_\ell^{-1}O & 0
    \end{bmatrix}.
    \label{eq:block_observable}
\end{equation}
Then a direct calculation gives
\begin{equation}
    Y_\ell^{(n)}
    =
    \langle\chi^{(\ell)}_K|
    \widehat{O}_\ell
    |\chi^{(\ell)}_K\rangle.
    \label{eq:Yell_block}
\end{equation}
Thus the coupled correction is obtained from a single observable on the
augmented state. The important point is that the scaling $\zeta_\ell$
appears inversely in $\widehat{O}_\ell$. Choosing $\zeta_\ell$ large at
fine levels makes the measured observable small, while the augmented
state remains close to normalised.

\subsubsection{Norm control and shot-noise decay}

We choose
\begin{equation}
    \zeta_\ell = \frac{c}{\sqrt{\tau_\ell}},
    \qquad c > 0.
    \label{eq:zeta_choice}
\end{equation}
The following lemma records the two estimates that make the measurement
variance decay with the level.

\begin{lemma}[Scaled augmented moments and shot-noise bound]
\label{lemma:scaled_augmented_bounds}
Let $|e_K^{(\ell)}\rangle$ be the final fine--coarse difference state
generated by the index-sharing coupling at level $\ell$, and write
$\tau_\ell = \lambda t/N_\ell$. Suppose that the coupled difference
satisfies the moment bounds
\begin{equation}
    \mathbb{E}_{\mathbf{j}_n}\|e_K^{(\ell)}\|^2
    \le C_2\tau_\ell,
    \qquad
    \mathbb{E}_{\mathbf{j}_n}\|e_K^{(\ell)}\|^4
    \le C_4\tau_\ell^2,
    \label{eq:error_state_moment_bounds}
\end{equation}
where the constants $C_2$ and $C_4$ may depend on the fixed problem
parameter $\lambda t$, but not on the level $\ell$. Let $\zeta_\ell =
c/\sqrt{\tau_\ell}$, and define $|\chi_K^{(\ell)}\rangle$ and
$\widehat{O}_\ell$ by \eqref{eq:augmented_state}
and~\eqref{eq:block_observable}. Then
\begin{equation}
    \mathbb{E}_{\mathbf{j}_n} S_\ell(\mathbf{j}_n)
    \le 1 + c^2 C_2,
    \qquad
    S_\ell(\mathbf{j}_n) := \|\chi_K^{(\ell)}\|^2,
    \label{eq:expected_chi_norm_bound}
\end{equation}
and
\begin{equation}
    \mathbb{E}_{\mathbf{j}_n} S_\ell(\mathbf{j}_n)^2
    \le 1 + 2c^2C_2 + c^4C_4.
    \label{eq:expected_chi_norm_square_bound}
\end{equation}
Moreover,
\begin{equation}
    \|\widehat{O}_\ell\|
    \le
    \left(\zeta_\ell^{-1} + \zeta_\ell^{-2}\right)\|O\|
    =
    \mathcal{O}\!\left(\tau_\ell^{1/2}\right)\|O\|.
    \label{eq:Ohat_norm_bound}
\end{equation}
Consequently, the expected conditional shot-noise variance of the scaled
augmented estimator satisfies
\begin{equation}
    \mathbb{E}_{\mathbf{j}_n}
    \bigl[
        \operatorname{Var}_{\mathrm{shot}}(\mathbf{j}_n)
    \bigr]
    =
    \mathcal{O}(\tau_\ell)
    =
    \mathcal{O}(2^{-\ell}).
    \label{eq:expected_shot_noise_decay}
\end{equation}
\end{lemma}

\begin{proof}
Since $|\psi_K^{(\ell-1)}\rangle$ is normalised, the augmented norm is
\[
    S_\ell(\mathbf{j}_n)
    =
    \|\chi_K^{(\ell)}\|^2
    =
    1 + \zeta_\ell^2\|e_K^{(\ell)}\|^2.
\]
Taking expectation and using $\zeta_\ell^2 = c^2/\tau_\ell$ together
with \eqref{eq:error_state_moment_bounds}, we obtain
\[
    \mathbb{E}_{\mathbf{j}_n}S_\ell(\mathbf{j}_n)
    =
    1 + \frac{c^2}{\tau_\ell}
    \mathbb{E}_{\mathbf{j}_n}\|e_K^{(\ell)}\|^2
    \le
    1 + c^2C_2.
\]
Similarly,
\[
    S_\ell(\mathbf{j}_n)^2
    =
    1
    + 2\zeta_\ell^2\|e_K^{(\ell)}\|^2
    + \zeta_\ell^4\|e_K^{(\ell)}\|^4,
\]
and therefore
\[
    \mathbb{E}_{\mathbf{j}_n}S_\ell(\mathbf{j}_n)^2
    \le
    1 + 2c^2C_2 + c^4C_4.
\]
This proves \eqref{eq:expected_chi_norm_bound} and
\eqref{eq:expected_chi_norm_square_bound}.

For the observable bound, the block structure gives
\[
    \widehat{O}_\ell
    =
    \begin{bmatrix}
        \zeta_\ell^{-2} & \zeta_\ell^{-1} \\
        \zeta_\ell^{-1} & 0
    \end{bmatrix}
    \otimes O.
\]
Hence
\[
    \|\widehat{O}_\ell\|
    \le
    \left\|
    \begin{bmatrix}
        \zeta_\ell^{-2} & \zeta_\ell^{-1} \\
        \zeta_\ell^{-1} & 0
    \end{bmatrix}
    \right\|\|O\|
    \le
    \left(\zeta_\ell^{-2} + \zeta_\ell^{-1}\right)\|O\|.
\]
Since $\zeta_\ell = c/\sqrt{\tau_\ell}$, this gives
$\|\widehat{O}_\ell\| = \mathcal{O}(\tau_\ell^{1/2})\|O\|$.

It remains to bound the shot noise. The normalised augmented state is
\[
    |\widetilde{\chi}_K^{(\ell)}\rangle
    =
    \frac{|\chi_K^{(\ell)}\rangle}{\sqrt{S_\ell(\mathbf{j}_n)}}.
\]
The correction is recovered as
\[
    Y_\ell^{(n)}
    =
    S_\ell(\mathbf{j}_n)
    \langle\widetilde{\chi}_K^{(\ell)}|
    \widehat{O}_\ell
    |\widetilde{\chi}_K^{(\ell)}\rangle.
\]
Therefore the conditional shot variance of the scaled measurement is
bounded by
\[
    \operatorname{Var}_{\mathrm{shot}}(\mathbf{j}_n)
    \le
    S_\ell(\mathbf{j}_n)^2
    \|\widehat{O}_\ell\|^2.
\]
Taking expectation over the sampled index sequence and using
\eqref{eq:expected_chi_norm_square_bound} and
\eqref{eq:Ohat_norm_bound}, we obtain
\[
    \mathbb{E}_{\mathbf{j}_n}
    \bigl[
        \operatorname{Var}_{\mathrm{shot}}(\mathbf{j}_n)
    \bigr]
    \le
    \mathbb{E}_{\mathbf{j}_n}S_\ell(\mathbf{j}_n)^2
    \|\widehat{O}_\ell\|^2
    =
    \mathcal{O}(\tau_\ell).
\]
Since $\tau_\ell = \tau_0 2^{-\ell}$, this is $\mathcal{O}(2^{-\ell})$.
\end{proof}

\medskip

\subsubsection{Near-unitary structure and dilation implementation}
\label{subsec:near_unitary_dilation}

The augmented-state construction has two consequences. First, it makes
the measurement variance decay with the MLMC level, as required for the
complexity theorem. Second, it turns the level correction into a
near-unitary non-Hermitian evolution, whose non-unitary part is small
and has the same Pauli-sum structure as the original sampled Hamiltonian
terms. We now make both points explicit.

The augmented state is generally not exactly normalised. Define
\begin{equation}
    S_\ell(\mathbf{j}_n)
    :=
    \|\chi^{(\ell)}_K\|^2,
    \qquad
    |\widetilde{\chi}^{(\ell)}_K\rangle
    :=
    \frac{|\chi^{(\ell)}_K\rangle}
         {\sqrt{S_\ell(\mathbf{j}_n)}}.
    \label{eq:normalised_chi}
\end{equation}
Then the block-observable identity~\eqref{eq:Yell_block} becomes
\begin{equation}
    Y_\ell^{(n)}
    =
    S_\ell(\mathbf{j}_n)
    \,
    \langle\widetilde{\chi}^{(\ell)}_K|
    \widehat{O}_\ell
    |\widetilde{\chi}^{(\ell)}_K\rangle.
    \label{eq:Yell_normalised}
\end{equation}
By Lemma~\ref{lemma:scaled_augmented_bounds}, the scaled augmented norm
has bounded second moment in expectation over the sampled qDRIFT
sequence, while
\[
    \|\widehat{O}_\ell\|
    =
    \mathcal{O}(\tau_\ell^{1/2}).
\]
Consequently,
\[
    \mathbb{E}_{\mathbf{j}_n}
    [\operatorname{Var}_{\mathrm{shot}}(\mathbf{j}_n)]
    =
    \mathcal{O}(\tau_\ell)
    =
    \mathcal{O}(2^{-\ell}).
\]
Thus the augmented state is almost normalised, and the observable used
to extract the correction becomes smaller at finer levels. Consequently,
the conditional single-shot measurement variance satisfies
\begin{equation}
    \operatorname{Var}_{\mathrm{shot}}
    \le
    S_\ell(\mathbf{j}_n)^2\,
    \|\widehat{O}_\ell\|^2
    =
    \mathcal{O}(\tau_\ell)
    =
    \mathcal{O}(2^{-\ell}).
    \label{eq:shot_var_decays}
\end{equation}
The measurement variance therefore decays at the same rate as the MLMC
correction variance from the index-sharing coupling. Including both
qDRIFT path randomness and shot noise, the total per-level variance
remains
\begin{equation}
    \operatorname{Var}_{\mathrm{total},\ell}
    =
    \mathcal{O}(2^{-\ell}),
    \label{eq:total_var}
\end{equation}
which preserves the $\beta = 1$ variance-decay condition used in the
MLMC complexity theorem.

\medskip

The same estimates also explain why the augmented dynamics is a
near-unitary object. The state norm is preserved up to
$\mathcal{O}(\tau_\ell)$ over the full coupled trajectory, while the
off-diagonal block in the augmented update
\[
    \mathcal{W}^{(\ell)}_m
    =
    \begin{bmatrix}
        U_{f,m} & \zeta_\ell(U_{f,m} - U_{c,m}) \\
        0       & U_{c,m}
    \end{bmatrix}
\]
has size
\[
    \|\zeta_\ell(U_{f,m} - U_{c,m})\|
    =
    \mathcal{O}(\tau_\ell^{1/2}),
\]
because $U_{f,m} - U_{c,m} = \mathcal{O}(\tau_\ell)$ and $\zeta_\ell =
c/\sqrt{\tau_\ell}$. Hence the augmented update is a small non-unitary
perturbation of a block-diagonal unitary update.

\medskip

We emphasise that the block update $\mathcal{W}^{(\ell)}_m$ in
\eqref{eq:augmented_recursion} is not, by itself, a near-term quantum
circuit. It is an algebraic representation of the coupled fine--coarse
difference dynamics. In particular, because of its upper triangular
off-diagonal block, $\mathcal{W}^{(\ell)}_m$ is generally nonunitary
and cannot be implemented directly as a deterministic gate sequence on
the physical register. The purpose of introducing
$\mathcal{W}^{(\ell)}_m$ is instead to expose the small nonunitary
correction that must be embedded or measured. Next, we show that this
correction has a Hermitian dilation structure with Pauli-sum components
of the same type as the original sampled Hamiltonian terms.
Specifically, within the $m$-th coarse block, let
\[
    \Delta H_m := H_{b_m} - H_{a_m} = H_{j_{n,2m}} - H_{j_{n,2m-1}}.
\]
On the first half-step, the fine and coarse paths both evolve under
$H_{a_m} = H_{j_{n,2m-1}}$. In the qDRIFT time variable $s \in
[0,\tau]$, the augmented state satisfies
\begin{equation}
    \frac{d}{ds}
    |\chi^{(\ell)}_m(s)\rangle
    =
    -i
    \begin{bmatrix}
        H_{j_{n,2m-1}} & 0 \\
        0       & H_{j_{n,2m-1}}
    \end{bmatrix}
    |\chi^{(\ell)}_m(s)\rangle.
    \label{eq:augmented_first_half}
\end{equation}
This part is exactly unitary. On the second half-step, the fine path
evolves under $H_{b_m} = H_{j_{n,2m}}$, while the coarse path continues
under $H_{a_m} = H_{j_{n,2m-1}}$. The augmented state then satisfies
\begin{equation}
    \frac{d}{ds}
    |\chi^{(\ell)}_m(s)\rangle
    =
    -i
    \begin{bmatrix}
        H_{j_{n,2m}} & \zeta_\ell\Delta H_m \\
        0       & H_{j_{n,2m-1}}
    \end{bmatrix}
    |\chi^{(\ell)}_m(s)\rangle,
    \qquad
    s \in [0,\tau].
    \label{eq:augmented_second_half}
\end{equation}
Equivalently, using the local variable $r = s/\tau \in [0,1]$,
\begin{equation}
    \frac{d}{dr}
    |\chi^{(\ell)}_m(r)\rangle
    =
    -i\tau
    \begin{bmatrix}
        H_{j_{n,2m}} & \zeta_\ell\Delta H_m \\
        0       & H_{j_{n,2m-1}}
    \end{bmatrix}
    |\chi^{(\ell)}_m(r)\rangle.
    \label{eq:augmented_second_half_rescaled}
\end{equation}
The factor multiplying the off-diagonal block is
\[
    \tau\zeta_\ell = c\,\tau^{1/2},
\]
so the non-unitary part of each local step is small at fine levels.

To isolate this non-unitary part, decompose the block generator in
\eqref{eq:augmented_second_half_rescaled} into Hermitian and
anti-Hermitian contributions. Define
\begin{equation}
    \mathcal{H}^{(\ell)}_m
    :=
    \tau
    \begin{bmatrix}
        H_{j_{n,2m}} & \frac{\zeta_\ell}{2}\Delta H_m \\[1mm]
        \frac{\zeta_\ell}{2}\Delta H_m & H_{j_{n,2m-1}}
    \end{bmatrix},
    \label{eq:H_eff_augmented}
\end{equation}
and
\begin{equation}
    \mathcal{K}^{(\ell)}_m
    :=
    \frac{\tau\zeta_\ell}{2}
    \begin{bmatrix}
        0 & -i\Delta H_m \\[1mm]
        i\Delta H_m & 0
    \end{bmatrix}.
    \label{eq:K_eff_augmented}
\end{equation}
Both $\mathcal{H}^{(\ell)}_m$ and $\mathcal{K}^{(\ell)}_m$ are
Hermitian, and the second-half evolution can be written as
\begin{equation}
    \frac{d}{dr}
    |\chi^{(\ell)}_m(r)\rangle
    =
    \left(
        -i\mathcal{H}^{(\ell)}_m
        +
        \mathcal{K}^{(\ell)}_m
    \right)
    |\chi^{(\ell)}_m(r)\rangle.
    \label{eq:nonhermitian_decomposition}
\end{equation}
Since $\|H_{j_{n,2m-1}}\| \le 1$, $\|H_{j_{n,2m}}\| \le 1$, and
$\|\Delta H_m\| \le 2$, the non-Hermitian component obeys
\begin{equation}
    \|\mathcal{K}^{(\ell)}_m\|
    \le
    \tau\zeta_\ell
    =
    c\,\tau^{1/2}
    =
    \mathcal{O}(\tau_\ell^{1/2}).
    \label{eq:K_small_augmented}
\end{equation}
Thus the generator is a Hermitian Hamiltonian evolution plus a small
Hermitian gain/loss term. This provides another view of why the
augmented dynamics remains close to norm-preserving.

\medskip

The non-Hermitian form~\eqref{eq:nonhermitian_decomposition} can be
converted back to a unitary or measurement-based evolution on an
enlarged Hilbert space using standard dilation ideas. In the simplest
postselected dilation schemes, one introduces an additional ancilla
qubit and represents the non-Hermitian flow generated by
$-i\mathcal{H}^{(\ell)}_m + \mathcal{K}^{(\ell)}_m$ as a block of a
larger unitary evolution whose Hermitian generator contains an ancilla
coupling proportional to $\sigma_Y \otimes \mathcal{K}^{(\ell)}_m$
\cite{liu2021probabilistic,mao2023measurement}. Measurement-based
realisations of such nonunitary evolutions have been demonstrated as
near-term protocols, for example in imaginary-time evolution
\cite{mao2023measurement}. In the present setting, the same mechanism
applies at the level of each augmented qDRIFT correction block: the
desired nonunitary update is recovered by measuring or postselecting the
dilation ancilla. Since $\|\mathcal{K}^{(\ell)}_m\| =
\mathcal{O}(\tau_\ell^{1/2})$, the nonunitary component becomes
perturbatively small on fine levels.

More generally, the moment-matching dilation framework of
Ref.~\cite{li2025momentmatching} provides a systematic way to embed
linear non-Hermitian dynamics of the form
\[
    \dot{x} = (-i\mathcal{H} + \mathcal{K})x
\]
into a unitary evolution on an enlarged space with tunable
encoding--evolution--evaluation maps. Such dilations can be used to
improve the postselection success probability and to control the
overhead associated with implementing the augmented qDRIFT correction.
This makes the present construction compatible with recent unitary
embedding approaches for general non-unitary linear dynamics.

Finally, the dilation preserves the algebraic structure needed for
qDRIFT-type implementation. For electronic-structure Hamiltonians and
spin models, the sampled Hamiltonian terms $H_{j_{n,2m-1}}$ and
$H_{j_{n,2m}}$ are Pauli strings, so
\[
    \Delta H_m = H_{j_{n,2m}} - H_{j_{n,2m-1}}
\]
is a short Pauli sum. The Hermitian and non-Hermitian parts above can be
written explicitly using the block-qubit Pauli operators as
\begin{align}
    \mathcal{H}^{(\ell)}_m
    &=
    \tau\left[
        |0\rangle\!\langle 0|_b \otimes H_{j_{n,2m}}
        +
        |1\rangle\!\langle 1|_b \otimes H_{j_{n,2m-1}}
        +
        \frac{\zeta_\ell}{2}
        X_b \otimes \Delta H_m
    \right],
    \label{eq:H_eff_pauli_structure}
    \\
    \mathcal{K}^{(\ell)}_m
    &=
    \frac{\tau\zeta_\ell}{2}
    Y_b \otimes \Delta H_m.
    \label{eq:K_eff_pauli_structure}
\end{align}
Here $b$ denotes the block qubit distinguishing the error and coarse
components. A further dilation only adds Pauli operators on the
dilation ancilla, such as $\sigma_Y$. Therefore the dilated Hamiltonians
remain sums of Pauli strings. In this sense, the augmented-state and
dilation construction does not leave the qDRIFT setting: it enlarges
the register and modifies the sampled Pauli terms, but the
implementation primitives remain Hamiltonian terms of the same type as
those used in the original randomised simulation.

\section{Gate Complexity and Crossover Analysis}
\label{sec:complexity}

We summarize the gate-complexity consequences of the MLMC-qDRIFT
construction developed in \cref{sec:mlmc} and identify the precision
regime in which the multilevel method becomes advantageous.  Full
derivations are given in \cref{app:complexity}.

Throughout this section we assume $\|O\|\le 1$, so a single-shot
measurement of $O$ has variance at most one.  We denote by
$\sigma^2\le 1$ the actual single-sample variance of standard qDRIFT,
including both qDRIFT path randomness and measurement shot noise.

We use the qDRIFT bias bound
\begin{equation}
    \bigl|\mathbb{E}[P_N]-P\bigr|
    \le
    \frac{B}{N},
    \qquad
    B = \mathcal{O}(t^2 \lambda^2),
    \label{eq:bias_constant}
\end{equation}
which follows from the diamond-norm bound
$\|\mathcal{E}^N - \mathcal{U}(t)\|_\diamond \le 2\lambda^2t^2/N$
of~\cite{campbell2019random}, giving bias decay exponent $\alpha = 1$.
As before, we split the mean-square error budget evenly:
$\mathrm{Bias}^2 \le \varepsilon^2/2$ and $\mathrm{Var} \le \varepsilon^2/2$.

\subsection{Standard qDRIFT}

The bias condition fixes the circuit depth to
$N_{\mathrm{std}} = \lceil \sqrt{2}B/\varepsilon \rceil$,
and the variance condition requires
$n_{\mathrm{std}} = \lceil 2\sigma^2/\varepsilon^2 \rceil$
independent samples.  The total gate count is therefore
\begin{equation}
    C_{\mathrm{std}}(\varepsilon)
    =
    N_{\mathrm{std}}\,n_{\mathrm{std}}
    =
    \mathcal{O}(\varepsilon^{-3}),
    \label{eq:std_cost}
\end{equation}
arising from the $\mathcal{O}(\varepsilon^{-1})$ depth needed to control
bias multiplied by the $\mathcal{O}(\varepsilon^{-2})$ samples needed to
control variance.

\subsection{MLMC-qDRIFT}

The index-sharing coupling gives total level-$\ell$ variance
$V_\ell \le A/N_\ell$ for $\ell \ge 1$, with cost per sample
$C_\ell = \frac{3}{2}N_\ell$, so the variance and cost exponents satisfy
$\beta = \gamma = 1$.  This is the boundary case of the MLMC complexity
theorem~\cite{giles2008multilevel}, which yields a logarithmic
correction (see \cref{app:complexity} for the full derivation):
\begin{equation}
    C_{\mathrm{MLMC}}(\varepsilon)
    =
    \mathcal{O}\!\left(\varepsilon^{-2}\log^2(1/\varepsilon)\right).
    \label{eq:mlmc_cost}
\end{equation}

\subsection{Crossover}
\label{sec:crossover}

The asymptotic improvement over standard qDRIFT is
\begin{equation}
    \frac{C_{\mathrm{std}}(\varepsilon)}{C_{\mathrm{MLMC}}(\varepsilon)}
    =
    \mathcal{O}\!\left(\frac{\varepsilon^{-1}}{\log^2(1/\varepsilon)}\right),
    \label{eq:speedup}
\end{equation}
so MLMC-qDRIFT is eventually cheaper for sufficiently small $\varepsilon$.
However, the multilevel estimator carries a coarse-level overhead that
can delay the crossover precision $\varepsilon^*$ at which the improved
scaling becomes visible, particularly for systems with large effective
one-norm $\lambda t$.  A detailed analysis of the two limiting crossover
regimes and the dependence of $\varepsilon^*$ on the problem constants
$A$, $B$, $N_0$, $V_0$, and $\sigma^2$ is given in
\cref{app:crossover}; the numerical crossover observed in
\cref{sec:numerics} is consistent with these predictions.

\section{Numerical Experiments}
\label{sec:numerics}

We validate the theoretical contributions of the paper through exact
and randomised numerical simulation on a concrete quantum spin system.
Three figures are presented. Figure~\ref{fig:qdrift_variance} examines
the variance and mean decay of the MLMC estimator using the averaged
qDRIFT channel, validating the $O(N_\ell^{-1})$ bias structure.
Figure~\ref{fig:shot_noise} samples actual random qDRIFT circuits via
the index-sharing coupling and
validates Lemma~\ref{lemma:index_sharing} by confirming that the
shot-noise variance of the augmented estimator of
Section~\ref{sec:augmented} decays as $O(\tau_\ell) = O(2^{-\ell})$.
Figure~\ref{fig:gate_count} validates the gate complexity improvement
of Theorem~\ref{thm:index_sharing_complexity}. Two complementary
computational approaches are used, and it is important to distinguish
them clearly. The first computes the level-$\ell$ measurement
probability $p_\ell = \operatorname{Tr}(\Pi_+
\mathcal{E}_\ell^{N_\ell}[\rho_0])$ exactly by iterating the averaged
qDRIFT channel via density-matrix evolution; no random circuits are
sampled and no measurement outcomes are drawn. Here we consider the
Pauli observable $O = Z_0$, and
\[
    \Pi_+ := \frac{I + O}{2}.
\]
Thus $p_\ell$ is the probability of obtaining the $+1$ outcome when
measuring $O$ after the averaged qDRIFT channel at level $\ell$. The
MLMC statistics for Figures~\ref{fig:qdrift_variance}
and~\ref{fig:gate_count} are then derived analytically from these exact
$p_\ell$ values via the Bernoulli
formulas~\eqref{eq:var_fine}--\eqref{eq:mean_diff}. The second
approach, used for Figure~\ref{fig:shot_noise}, samples actual random
qDRIFT circuits via the index-sharing coupling of
Section~\ref{subsec:index_sharing} and directly computes the shot-noise
variance of the augmented estimator of Section~\ref{sec:augmented} on
each realisation.

All experiments are performed on the one-dimensional Heisenberg XYZ
spin chain with $n = 6$ qubits,
\begin{equation}
    H \;=\; \sum_{j=0}^{n-2}
        \Bigl(
          J_x\,X_j X_{j+1}
          + J_y\,Y_j Y_{j+1}
          + J_z\,Z_j Z_{j+1}
        \Bigr),
    \qquad
    J_x = 1.0,\quad J_y = 0.5,\quad J_z = 0.8,
    \label{eq:heisenberg}
\end{equation}
with evolution time $T = 1$. The Hamiltonian decomposes into $M = 15$
Hermitian two-body terms with one-norm $\lambda \approx 11.5$, giving
$\lambda^2T^2 \approx 132$. The reference value $P_{\mathrm{exact}} =
\langle Z_0\rangle_T \approx 0.5024$ is obtained by exact matrix
exponentiation of $H$, with corresponding measurement probability
$p_\infty \approx 0.7512$. Since the single-shot measurement of $O =
Z_0$ yields a $\pm 1$ Bernoulli outcome with parameter $p_\ell$, the
single-shot variance at level $\ell$ is $\sigma^2 =
4p_\ell(1-p_\ell)$, which converges to $4p_\infty(1-p_\infty) \approx
0.748$ as $\ell \to \infty$, specialising the general bound $\sigma^2
\leq \|O\|^2 = 1$ of Section~\ref{sec:complexity}.

\subsection{MLMC Variance and Mean Decay}
\label{sec:var_decay}

\begin{figure}[htp]
    \centering
    \includegraphics[width=\linewidth]{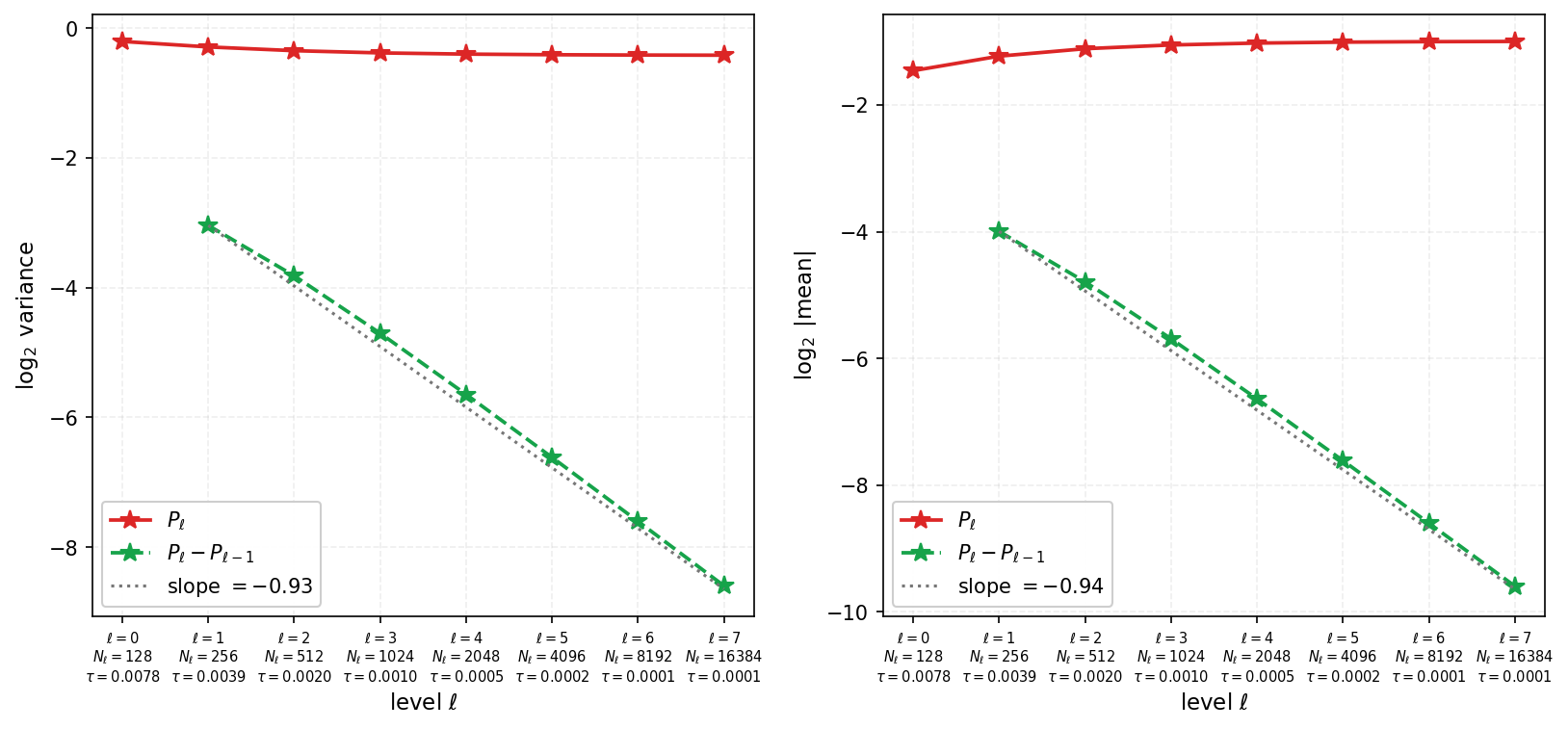}
    \caption{MLMC variance (left) and absolute mean (right) for the
      qDRIFT estimator on the 6-qubit Heisenberg XYZ
      chain~\eqref{eq:heisenberg}, plotted on a $\log_2$ scale against
      MLMC level $\ell$. All quantities are computed analytically from
      the exact channel probabilities $p_\ell$ via the Bernoulli
      formulas~\eqref{eq:var_fine}--\eqref{eq:mean_diff}; no random
      circuits are sampled. Red solid lines show the single-level
      estimator $P_\ell$; green dashed lines show the coupled correction $\Delta P_\ell = P_\ell - P_{\ell-1}$, where $p_\ell$ and $p_{\ell-1}$ are computed from the index-sharing qDRIFT channel and the difference is evaluated via the correlated Bernoulli formula~\eqref{eq:var_diff}; dotted lines show the least-squares
      reference slope. The flat red line confirms that
      $\operatorname{Var}(P_\ell) = 4p_\ell(1-p_\ell) \approx 0.75$
      is level-independent, the shot-noise floor of a single projective
      measurement of $Z_0$. The decaying green dashed line has fitted
      rate $\hat{\beta} \approx 0.93$, consistent with the
      $O(N_\ell^{-1})$ qDRIFT bias bound $|p_\ell - p_{\ell-1}| =
      O(\lambda^2T^2/N_\ell)$; this decay reflects the averaged channel
      structure rather than the index-sharing coupling, which is
      validated in Figure~\ref{fig:shot_noise}. The right panel shows
      the absolute mean decaying at fitted rate $\hat{\alpha} \approx
      0.94$, confirming first-order bias cancellation in the MLMC
      telescoping sum~\eqref{eq:telescoping}.}
    \label{fig:qdrift_variance}
\end{figure}

Figure~\ref{fig:qdrift_variance} shows $\log_2$ of the per-level
variance (left panel) and absolute mean (right panel) for the qDRIFT
estimator. All four quantities are computed analytically from the exact
$p_\ell$ values, no random circuits are sampled and no measurement
outcomes are drawn:
\begin{align}
    \operatorname{Var}(P_\ell)
        &= 4\,p_\ell(1-p_\ell), \label{eq:var_fine} \\
    |\mathbb{E}[P_\ell]|
        &= |2p_\ell - 1|, \label{eq:mean_fine} \\
    \operatorname{Var}(\Delta P_\ell)
        &= 4|p_\ell - p_{\ell-1}|
         \cdot (1 - |p_\ell - p_{\ell-1}|), \label{eq:var_diff} \\
    |\mathbb{E}[\Delta P_\ell]|
        &= 2|p_\ell - p_{\ell-1}|. \label{eq:mean_diff}
\end{align}
The correction formula~\eqref{eq:var_diff} corresponds to a maximally
coupled Bernoulli measurement of two levels. If $X_\ell, X_{\ell-1}
\in \{\pm 1\}$ have success probabilities $p_\ell$ and $p_{\ell-1}$
and are coupled via the same uniform random variable $u \sim
\mathcal{U}[0,1]$, then $\mathbb{P}(X_\ell \neq X_{\ell-1}) =
|p_\ell - p_{\ell-1}|$, so $(X_\ell - X_{\ell-1})^2 = 4$ with
probability $|p_\ell - p_{\ell-1}|$ and $0$ otherwise, giving
\begin{equation*}
    \operatorname{Var}(X_\ell - X_{\ell-1})
    = 4|p_\ell - p_{\ell-1}|\bigl(1 - |p_\ell - p_{\ell-1}|\bigr).
\end{equation*}
This coupling acts only at the level of measurement outcomes; it does
not involve the random index sequences $\mathbf{j}$ that drive the
qDRIFT circuits. A more complete treatment, which couples the two
levels coherently through a shared circuit and measures their difference
in a single shot, is provided by the augmented state method of
Section~\ref{sec:augmented}, validated in
Section~\ref{sec:shot_noise}.

The coarsest level is set to $N_0 = 128 \approx \lambda^2T^2 \approx
132$, placing the simulation at the threshold of the asymptotic
$O(N_\ell^{-1})$ convergence regime. The gate counts $N_\ell = N_0
\cdot 2^\ell$ span $\{128, 256, \ldots, 16384\}$ for $\ell \in
\{0,\ldots,7\}$.

Two properties are verified numerically. First, the analytically
computed single-level variance $\operatorname{Var}(P_\ell) =
4p_\ell(1-p_\ell) \approx 0.75$ is flat across all levels (red solid
line, left panel): since $p_\ell \to p_\infty$ monotonically, the
Bernoulli variance converges to $4p_\infty(1-p_\infty) \approx 0.748$
and is level-independent. This is the shot-noise floor that any single
projective measurement of $Z_0$ would incur, regardless of circuit
depth. Second, the analytically computed correction variance (green
dashed line, left panel) decays geometrically with fitted rate
$\hat{\beta} = 0.93$, consistent with the $O(N_\ell^{-1})$ qDRIFT
bias bound, which gives $|p_\ell - p_{\ell-1}| = O(\lambda^2T^2/
N_\ell)$. The absolute mean of the correction (green dashed line, right
panel) decays at fitted rate $\hat{\alpha} = 0.94$, confirming
first-order bias cancellation in the MLMC telescoping
sum~\eqref{eq:telescoping}. The slight departures of $\hat{\beta}$ and
$\hat{\alpha}$ from unity are attributable to the finite regression
window of eight levels and do not affect the asymptotic complexity bound.

It is important to note that this decay is a consequence of the averaged
channel structure — specifically that $|p_\ell - p_{\ell-1}|$ shrinks
as $N_\ell$ grows — and does not by itself validate
Lemma~\ref{lemma:index_sharing}, which concerns the stochastic variance
$V_\ell = \operatorname{Var}_{\mathbf{j}}[Y_\ell(\mathbf{j})]$ across
random index sequences and requires the index-sharing coupling
specifically. That validation is provided by
Figure~\ref{fig:shot_noise} in Section~\ref{sec:shot_noise}.

\subsection{Shot-Noise Variance of the Augmented Estimator}
\label{sec:shot_noise}

\begin{figure}[htp]
    \centering
    \includegraphics[width=0.7\linewidth]{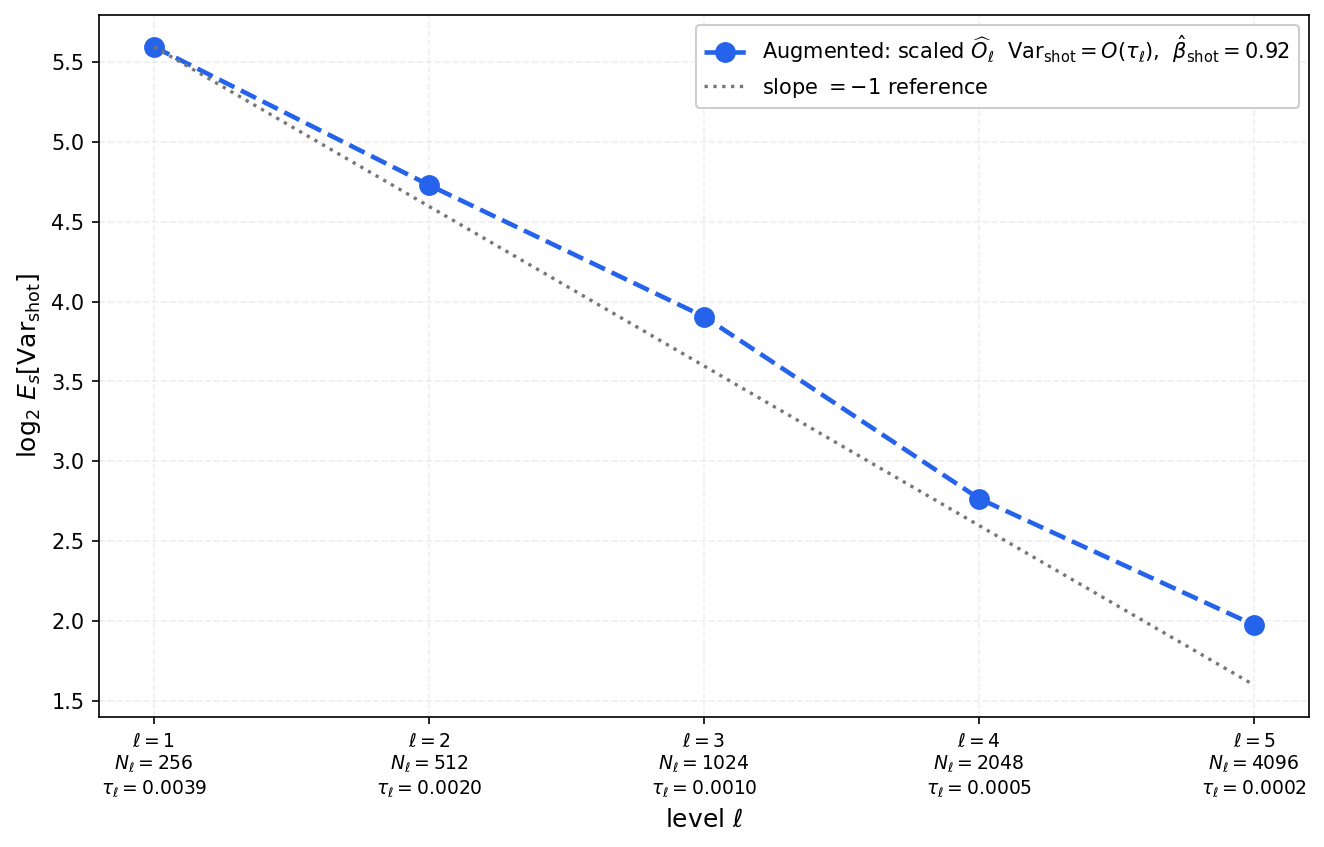}
    \caption{Mean conditional shot-noise variance of the augmented
      estimator,
      \(\mathbb{E}_{\mathbf{s}}[\operatorname{Var}_{\mathrm{shot}}]\),
      plotted on a \(\log_2\) scale against correction level \(\ell\)
      for the 6-qubit Heisenberg XYZ chain~\eqref{eq:heisenberg}.
      Each point is estimated by averaging~\eqref{eq:var_shot_sample}
      over \(n_\ell\in\{80,\ldots,300\}\) independently sampled
      index-sharing circuits at gate counts
      \(N_\ell\in\{256,\ldots,4096\}\), using
      \(\zeta_\ell=1/\sqrt{\tau_\ell}\).  The fitted decay rate is
      \(\hat\beta_{\mathrm{shot}}\approx 1.0\), matching the reference
      slope \(-1\).  Thus the shot-noise contribution decays as
      \(\mathcal{O}(\tau_\ell)=\mathcal{O}(2^{-\ell})\), as predicted
      by~\cref{eq:shot_var_decays}.}
    \label{fig:shot_noise}
\end{figure}

Figure 3 samples actual random qDRIFT circuits via the index-sharing
coupling and validates the measurement-level variance scaling of the
augmented estimator from Section~2.3: the scaled augmented observable
\(\widehat O_\ell\) reduces the conditional shot-noise variance from a
level-independent \(\mathcal{O}(1)\) floor to
\(\mathcal{O}(\tau_\ell)=\mathcal{O}(2^{-\ell})\).  This experiment uses
actual random qDRIFT circuits sampled with the index-sharing coupling,
rather than averaged-channel probabilities.

Speifically, for each level \(\ell\in\{1,\ldots,5\}\), we sample
\(n_\ell\in\{80,\ldots,300\}\) random index sequences
\(\mathbf{s}=(s_1,\dots,s_{2K})\).  The fine and coarse states are
propagated using~\eqref{eq:fine_unitary}--\eqref{eq:coarse_unitary}, and
we form the difference state
\(|e^{(\ell)}_K\rangle=|\psi^{(\ell)}_{2K}\rangle-
|\psi^{(\ell-1)}_K\rangle\).  With
\(\zeta_\ell=c/\sqrt{\tau_\ell}\) and \(c=1\), the augmented norm is
\(S_\ell(\mathbf{s})=1+\zeta_\ell^2\|e^{(\ell)}_K\|^2\).  For each fixed
sample path, the conditional shot-noise variance is
\begin{equation}
    \operatorname{Var}_{\mathrm{shot}}(\mathbf{s})
    =
    S_\ell(\mathbf{s})
    \langle\chi^{(\ell)}_K|
    \widehat O_\ell^2
    |\chi^{(\ell)}_K\rangle
    -
    Y_\ell(\mathbf{s})^2 .
    \label{eq:var_shot_sample}
\end{equation}
For \(O=Z_0\), \(O^2=I\), and the block structure gives the exact
identity
\[
    \langle\chi^{(\ell)}_K|
    \widehat O_\ell^2
    |\chi^{(\ell)}_K\rangle
    =
    \|e^{(\ell)}_K\|^2+\zeta_\ell^{-2}
    =
    \|e^{(\ell)}_K\|^2+\tau_\ell/c^2.
\]
This follows from the normalization of the fine and coarse states, which
implies
\(2\operatorname{Re}\langle e^{(\ell)}_K|\psi^{(\ell-1)}_K\rangle
=-\|e^{(\ell)}_K\|^2\).  We then average
\(\operatorname{Var}_{\mathrm{shot}}(\mathbf{s})\) over the sampled
index sequences.

The fitted rate \(\hat\beta_{\mathrm{shot}}\approx 1.0\) confirms that
\(\mathbb{E}_{\mathbf{s}}[\operatorname{Var}_{\mathrm{shot}}]
=\mathcal{O}(2^{-\ell})\).  Together with the index-sharing bound
\(V_\ell^{\mathrm{path}}=\mathcal{O}(2^{-\ell})\) from
Lemma~\ref{lemma:index_sharing}, this gives the total level-variance
decay \(V_\ell^{\mathrm{total}}=\mathcal{O}(2^{-\ell})\), which is the
\(\beta=1\) condition used in Theorem~\ref{thm:index_sharing_complexity}.
\subsection{Gate Complexity Validation}
\label{sec:gate_numerical}

\begin{figure}[t]
    \centering
    \includegraphics[width=0.8\linewidth]{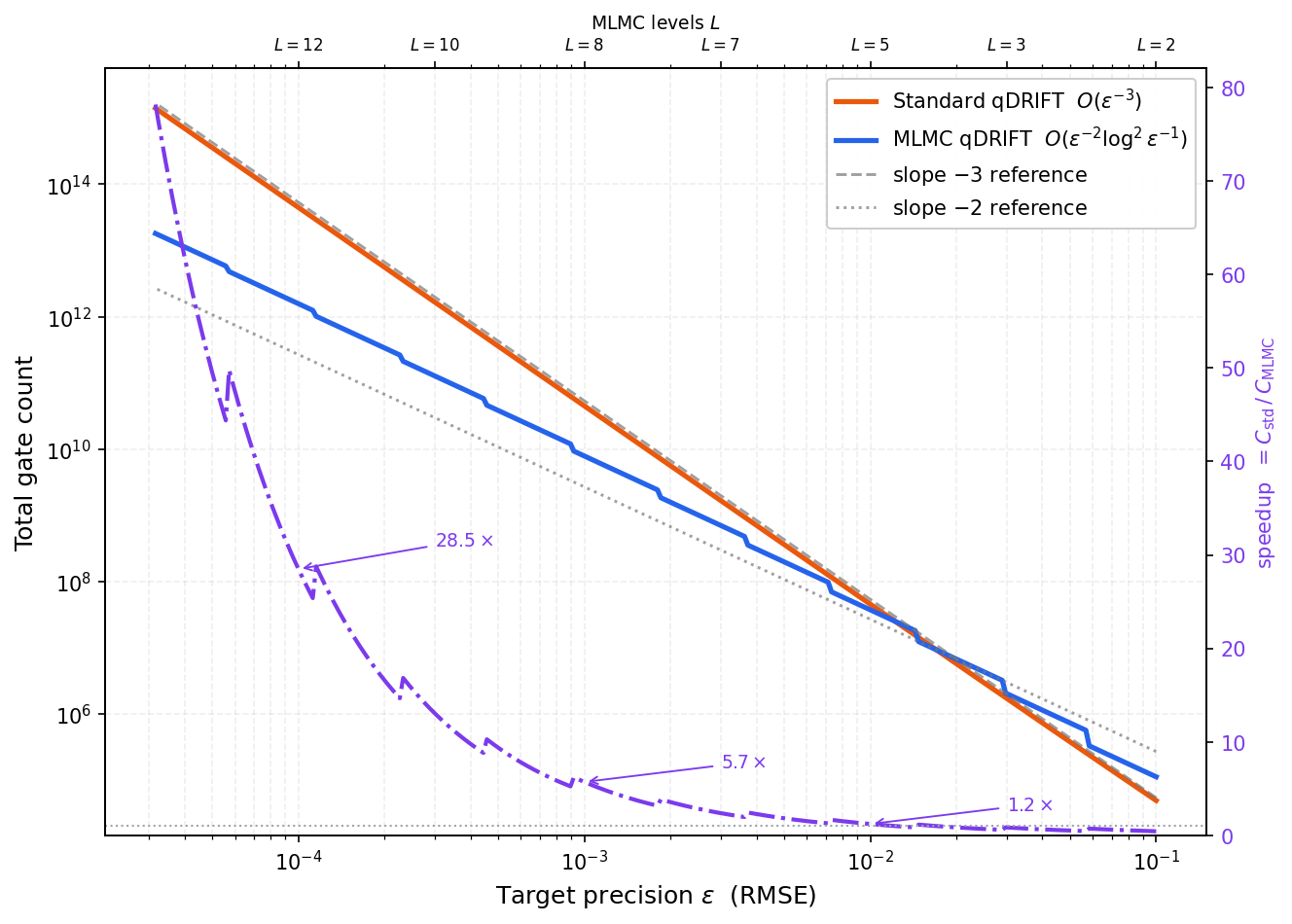}
    \caption{Total gate complexity as a function of target RMSE
      $\varepsilon$ for standard qDRIFT (orange) and MLMC qDRIFT
      (blue), computed from the cost
      formulae~\eqref{eq:std_cost}--\eqref{eq:mlmc_cost} instantiated
      for $O = Z_0$ with system-specific values $c_1 \approx 10.55$,
      $N_0 = 128$, and $\sigma^2 = 4p_{N_L}(1-p_{N_L})$. Dashed and
      dotted gray lines show the slope $-3$ and slope $-2$ references,
      both anchored to the MLMC curve at $\varepsilon = 0.02$. The
      right axis (purple, dash-dot) shows the gate-count speedup
      $C_{\mathrm{std}}/C_{\mathrm{MLMC}}$; the horizontal dotted line
      marks break-even. The crossover occurs near $\varepsilon^* \approx
      0.02$, consistent with the log-dominated
      estimate~\eqref{eq:crossover_log}. The upper axis shows the
      number of MLMC levels $L$ required at each precision.}
    \label{fig:gate_count}
\end{figure}

Figure~\ref{fig:gate_count} validates the cost comparison of
Section~\ref{sec:complexity} for the six-qubit Heisenberg example.  The
plot uses the standard qDRIFT cost~\eqref{eq:std_cost} and the
MLMC-qDRIFT cost~\eqref{eq:mlmc_cost}, with the MLMC Giles sum
\(\mathcal{S}_L=\sum_{\ell=0}^L\sqrt{V_\ell C_\ell}\) evaluated using
the level variances and costs from the preceding numerical experiments.
For \(O=Z_0\), the measurement probability satisfies
\[
    \langle Z_0\rangle = 2p-1.
\]
We therefore fit the probability bias as
\[
    |p_\ell-p_\infty|\approx \frac{c_p}{N_\ell},
\]
using levels \(\ell\in\{3,\ldots,7\}\), where the
\(\mathcal{O}(N_\ell^{-1})\) regime is cleanly observed.  This gives
\(c_p\approx 10.55\), corresponding to the observable-bias constant
\[
    B_Z \approx 2c_p .
\]
The single-sample variance for standard qDRIFT is
\[
    \sigma^2 = 4p_{N_L}(1-p_{N_L})\le 1,
\]
with \(p_{N_L}\) approximated by
\[
    p_{N_L}\approx p_\infty-\frac{c_p}{N_L},
    \qquad
    p_\infty\approx 0.7512 .
\]

For this system, \(\lambda^2T^2\approx 132\) and \(N_0=128\), so the
coarsest level is chosen near the natural scale suggested by the qDRIFT
bias bound.  This places the example near the transition between the
overhead-dominated and log-dominated crossover regimes discussed in
Section~\ref{sec:complexity}.  The observed crossover in
Figure~\ref{fig:gate_count} occurs near
\[
    \varepsilon^*\approx 0.02 .
\]
Below this precision, the multilevel allocation increasingly favors
cheap coarse samples and uses only a small number of expensive fine-level
corrections.  The resulting gate-count reductions are approximately
\(1.2\times\), \(5.7\times\), and \(28\times\) at
\(\varepsilon=10^{-2}\), \(10^{-3}\), and \(10^{-4}\), respectively.
This trend is consistent with the asymptotic speedup factor
\[
    \frac{C_{\mathrm{std}}(\varepsilon)}
         {C_{\mathrm{MLMC}}(\varepsilon)}
    =
    \mathcal{O}\!\left(
        \frac{\varepsilon^{-1}}{\log^2(1/\varepsilon)}
    \right),
\]
which follows from comparing
\(C_{\mathrm{std}}=\mathcal{O}(\varepsilon^{-3})\) with
\(C_{\mathrm{MLMC}}=\mathcal{O}(\varepsilon^{-2}\log^2(1/\varepsilon))\).

%══════════════════════════════════════════════════════════════════════════════

\section{Summary and Discussion}
\label{sec:discussion}

Multilevel Monte Carlo has been one of the notable successes of modern
computational probability: by coupling approximations at different
resolutions, it converts a costly single-level Monte Carlo estimator into
a hierarchy in which most samples are taken at low cost and only a few
are taken at high accuracy.  In this work we have adapted this principle
to randomized quantum simulation.  For qDRIFT-based observable
estimation, we constructed a hierarchy of circuit depths and coupled
adjacent levels by sharing random Hamiltonian-term samples.  This
index-sharing coupling yields a decaying variance for the level
corrections and reduces the total gate complexity from the standard
qDRIFT scaling
\(
    \mathcal{O}(\varepsilon^{-3})
\) to \( 
    \mathcal{O}\!\left(\varepsilon^{-2}\log^2(1/\varepsilon)\right),
\)
while retaining qDRIFT's absence of explicit dependence on the number of
Hamiltonian terms. To faciliate near-term implementations, We also introduced a scaled augmented-state formulation that makes the
quantum measurement shot noise decay at the same rate, thereby preserving
the multilevel variance reduction at the circuit measurement level. The numerical experiments confirm the predicted
variance decay and show that the asymptotic improvement can translate
into substantial gate-count savings at high precision.

The framework developed here is compatible with other techniques for
reducing the cost of randomized simulation.  For Hamiltonians with a few
large terms and many small terms, composite deterministic--randomized
partitioning methods~\cite{jin2025partially,hagan2023composite} could be
combined with MLMC-qDRIFT so that the multilevel estimator acts only on a
reduced random component with a smaller effective one-norm.  Similarly,
bias-reduction methods such as Richardson extrapolation, including those
used in qFLO~\cite{watson2024randomly}, are complementary to the present
variance-reduction approach.  In classical numerical analysis, multilevel
Monte Carlo has been successfully combined with Richardson--Romberg
extrapolation \cite{lemaire2017multilevel}; the same idea suggests a possible route to quantum
algorithms that reduce both the qDRIFT bias and the sampling variance.

More broadly, the significance of MLMC-qDRIFT is not limited to
Hamiltonian simulation.  Randomization is increasingly used as a design
principle in quantum algorithms, especially when circuit depth,
ancilla-assisted coherent control, or block-encoding overheads are
expensive.  Examples include randomized simulation of Lindblad dynamics
and thermal-state preparation~\cite{chen2025randomized}, randomized
adiabatic algorithms for quantum linear systems~\cite{subasi2019linear,
jennings2025randomized}, and randomized algorithms for sampling matrix
functions and solving linear-algebra problems~\cite{wang2024qubit}.
These algorithms all introduce statistical error in exchange for reduced
coherent quantum resources.  A natural direction for future work is to
develop multilevel couplings for this broader class of randomized quantum
algorithms.  Whenever a quantum algorithm admits a hierarchy of
approximations with increasing accuracy and cost, and whenever adjacent
levels can be coupled so that correction variances decay sufficiently
fast, MLMC-type variance reduction may provide a systematic way to turn
randomization into a precision-efficient computational tool.
\section{Acknowledgements}
This research is supported by the NSF Grants No. DMS-2111221 and No. CCF-2312456.

\bibliographystyle{plain}
\bibliography{references}

\appendix
\begin{center}
    \vspace{3em}
    \textbf{\Large Appendix}
\end{center}

\tableofcontents 

\renewcommand{\appendixname}{APPENDIX}
\renewcommand{\thesection}{\Alph{section}}
\renewcommand{\thesubsection}{\thesection.\arabic{subsection}}
\renewcommand{\thesubsubsection}{\thesubsection.\arabic{subsubsection}}
\makeatletter
\renewcommand{\p@subsection}{}
\renewcommand{\p@subsubsection}{}
\renewcommand{\theequation}{\arabic{subsection}.\arabic{equation}}
\makeatother

\renewcommand{\thefigure}{\arabic{figure}}
\renewcommand{\thetable}{\arabic{table}}
\renewcommand{\figurename}{Figure}
\setcounter{figure}{0}
\setcounter{equation}{0}
\setcounter{secnumdepth}{3}

\section{Multilevel Monte Carlo (MLMC) Methods}
\label{app:mlmc_methods}

The standard Monte Carlo (MC) method estimates the expectation
$\mathbb{E}[P]$ of a random variable by computing the empirical mean
of $n$ independent samples. While robust, this approach can be
computationally prohibitive when the samples themselves are based on
an expensive numerical approximation. If $P_L$ denotes a
high-accuracy (but costly) approximation of $P$, the Mean Squared
Error (MSE) scales as $\mathcal{O}(n^{-1})$, requiring a massive number of
samples to reduce discretization bias to an acceptable level.

MLMC circumvents this by defining a hierarchy of levels
$\ell = 0, 1, \dots, L$, where each level represents a numerical
approximation of increasing accuracy and cost.

\subsection{The Multilevel Identity}
The core innovation of MLMC is to exploit the linearity of
expectation. Instead of estimating $\mathbb{E}[P_L]$ directly, we
express it as a telescoping sum of expectations of differences:
\begin{equation}
    \mathbb{E}[P_L]
    = \mathbb{E}[P_0]
      + \sum_{\ell=1}^L \mathbb{E}[P_\ell - P_{\ell-1}].
\end{equation}
To estimate this efficiently, we construct an independent empirical
estimator $Y_\ell$ for each level. For the initial coarse level,
$Y_0$ estimates $\mathbb{E}[P_0]$. For subsequent levels
($\ell \ge 1$), $Y_\ell$ estimates the increment
$\mathbb{E}[P_\ell - P_{\ell-1}]$ using $n_\ell$ independent samples:
\begin{equation}
    Y_\ell
    = \frac{1}{n_\ell} \sum_{i=1}^{n_\ell}
      \bigl( P_\ell^{(i)} - P_{\ell-1}^{(i)} \bigr).
\end{equation}
The efficiency of this method relies on \textit{coupling}: if
$P_\ell$ and $P_{\ell-1}$ are simulated using the same underlying
random variables such that they are highly correlated, the variance
of a single sample difference $V_\ell = \operatorname{Var}(P_\ell -
P_{\ell-1})$ can be significantly smaller than the variance of the
individual terms, allowing fewer samples $n_\ell$ at the
computationally expensive fine levels.

\subsection{Complexity and Convergence}
The total cost of the MLMC estimator $Y = \sum_{\ell=0}^L Y_\ell$
is
\begin{equation}
    \mathrm{Cost}_{\mathrm{total}}
    = \sum_{\ell=0}^L n_\ell\, \mathcal{C}_\ell,
\end{equation}
where $\mathcal{C}_\ell$ is the expected cost per sample at level
$\ell$. The following theorem formalises the complexity gains
achieved when the variance of the differences decays faster than
the cost increases.

\begin{theorem}[\textbf{MLMC Complexity Theorem,
\cite{giles2015multilevel}}]
\label{thm:mlmc_complexity}
Let $P$ denote a random variable and $P_\ell$ its level-$\ell$
numerical approximation. Suppose there exist independent estimators
$Y_\ell$ (where $Y_0$ is an unbiased estimator of $\mathbb{E}[P_0]$
and $Y_\ell$ is an unbiased estimator of $\mathbb{E}[P_\ell -
P_{\ell-1}]$). Let $V_\ell = \operatorname{Var}(P_\ell -
P_{\ell-1})$ be the variance of a single sample difference and
$\mathcal{C}_\ell$ its expected cost. Let $r$ be the refinement
factor between levels. If there exist positive constants $\alpha,
\beta, \gamma, c_1, c_2, c_3$ with $\alpha \ge \tfrac{1}{2}
\min(\beta,\gamma)$ and:
\begin{enumerate}
    \item \textbf{Bias decay:}
          $|\mathbb{E}[P_\ell - P]| \le c_1\, r^{-\alpha\ell}$
    \item \textbf{Variance decay:}
          $V_\ell \le c_2\, r^{-\beta\ell}$
    \item \textbf{Cost growth:}
          $\mathcal{C}_\ell \le c_3\, r^{\gamma\ell}$
\end{enumerate}
then there exists a positive constant $c_4$ such that for any
$\varepsilon < e^{-1}$ there are values $L$ and $\{n_\ell\}$ for
which the multilevel estimator $Y = \sum_{\ell=0}^L Y_\ell$ has
mean square error $\mathbb{E}[(Y - \mathbb{E}[P])^2] < \varepsilon^2$
with total complexity bounded by
\begin{equation}
    \mathbb{E}[C]
    \le \begin{cases}
        c_4\,\varepsilon^{-2}
          & \beta > \gamma, \\[2pt]
        c_4\,\varepsilon^{-2}(\log\varepsilon)^2
          & \beta = \gamma, \\[2pt]
        c_4\,\varepsilon^{-2-(\gamma-\beta)/\alpha}
          & \beta < \gamma.
    \end{cases}
\end{equation}
\end{theorem}

\section{Variance Bounds (Proof of Lemma)} \label{app:proof_of_lemma}

\textbf{Lemma \ref{lemma:index_sharing}.}
(Index-Sharing Variance Bound)
Let $P_\ell^{(\mathbf{j}_n,\ell)}$ and $P_{\ell-1}^{(\tilde{\mathbf{j}}_n,\ell)}$
be the random variables from the index-sharing coupling. Under
$\|H_j\| \le 1$ and $\|O\| \le 1$, the correction variance satisfies
\begin{equation}
    V_\ell
    =
    \operatorname{Var}\!\left(
        P_\ell^{(\mathbf{j}_n,\ell)}
        - P_{\ell-1}^{(\tilde{\mathbf{j}}_n,\ell)}
    \right)
    \le
    C_{\mathrm{IS}}\,\frac{\lambda^2t^2}{N_0}\,2^{-\ell},
    \label{eq:variance_bound_app}
\end{equation}
where $C_{\mathrm{IS}}$ is a universal constant independent of $\ell$,
$N_0$, and the number of Hamiltonian terms $M$.

\begin{proof}
Bounding the variance by the second moment, it is enough to estimate
$\mathbb{E}\left[|P_\ell^{(\mathbf{j}_n,\ell)}
- P_{\ell-1}^{(\tilde{\mathbf{j}}_n,\ell)}|^2\right]$.
Set
\[
    \tau := \tau_\ell, \qquad K := N_{\ell-1} = \frac{N_\ell}{2}.
\]
For the $k$-th coupled block, the two i.i.d.\ indices drawn from
$\{p_j\}_{j=1}^M$ are $j_{n,2k-1}$ and $j_{n,2k}$. Write
\[
    a_k := j_{n,2k-1}, \qquad b_k := j_{n,2k},
\]
and define
\[
    A_k := H_{a_k}, \qquad B_k := H_{b_k}.
\]
Thus, according to the index-sharing coupling
\eqref{eq:fine_unitary}--\eqref{eq:coarse_unitary},
\[
    F_k := e^{-i\tau B_k}e^{-i\tau A_k},
    \qquad
    C_k := e^{-i2\tau A_k}.
\]
Let
\[
    |\psi_k^f\rangle = F_k|\psi_{k-1}^f\rangle,
    \qquad
    |\psi_k^c\rangle = C_k|\psi_{k-1}^c\rangle,
    \qquad
    |\psi_0^f\rangle = |\psi_0^c\rangle = |\psi_0\rangle.
\]
At the final synchronised time,
\[
    P_\ell^{(\mathbf{j}_n,\ell)}
    = \langle\psi_K^f|O|\psi_K^f\rangle,
    \qquad
    P_{\ell-1}^{(\tilde{\mathbf{j}}_n,\ell)}
    = \langle\psi_K^c|O|\psi_K^c\rangle.
\]
Introducing the middle quadratic form
$\langle\psi_K^f|O|\psi_K^c\rangle$, we get
\begin{align}
    P_\ell^{(\mathbf{j}_n,\ell)} - P_{\ell-1}^{(\tilde{\mathbf{j}}_n,\ell)}
    &=
    \langle\psi_K^f|O
    \bigl(|\psi_K^f\rangle - |\psi_K^c\rangle\bigr)
    +
    \bigl(\langle\psi_K^f| - \langle\psi_K^c|\bigr)
    O|\psi_K^c\rangle.
\end{align}
Since $\|O\| \le 1$ and both states are normalised,
\begin{equation}
    |P_\ell^{(\mathbf{j}_n,\ell)} - P_{\ell-1}^{(\tilde{\mathbf{j}}_n,\ell)}|
    \le
    2\||\psi_K^f\rangle - |\psi_K^c\rangle\|.
    \label{eq:observable_to_state_error}
\end{equation}

It remains to bound this synchronised state difference. Define
\[
    e_m := |\psi_m^f\rangle - |\psi_m^c\rangle.
\]
Then
\begin{equation}
    e_m
    =
    F_m e_{m-1} + d_m,
    \qquad
    d_m := (F_m - C_m)|\psi_{m-1}^c\rangle.
    \label{eq:state_error_recursion_app}
\end{equation}
Because $F_m$ is unitary,
\begin{equation}
    \|e_m\|^2
    =
    \|e_{m-1}\|^2
    + 2\operatorname{Re}\langle e_{m-1}, F_m^\dagger d_m\rangle
    + \|d_m\|^2.
    \label{eq:error_square_recursion_app}
\end{equation}

We now estimate $d_m$ using the variation-of-constants formula. Since
\[
    C_m = e^{-i2\tau A_m} = e^{-i\tau A_m}e^{-i\tau A_m},
\]
we have
\begin{equation}
    F_m - C_m
    =
    \bigl(e^{-i\tau B_m} - e^{-i\tau A_m}\bigr)e^{-i\tau A_m}.
    \label{eq:block_defect_factorization}
\end{equation}

The variation-of-constants formula gives

\begin{equation}
    e^{-i\tau B_m} - e^{-i\tau A_m}
    =
    -i\int_0^\tau
    e^{-i(\tau-s)B_m}
    (B_m - A_m)
    e^{-isA_m}\,ds.   \label{eq:vocf_exp_difference}
\end{equation}
Hence
\begin{equation}
    d_m
    =
    -i\int_0^\tau
    e^{-i(\tau-s)B_m}
    (B_m - A_m)
    e^{-i(s+\tau)A_m}
    |\psi_{m-1}^c\rangle\,ds.
    \label{eq:dm_vocf}
\end{equation}

Since all exponential factors are unitary and $\|A_m\|, \|B_m\| \le 1$,
\begin{equation}
    \|d_m\|
    \le
    \int_0^\tau \|B_m - A_m\|\,ds
    \le
    2\tau.
    \label{eq:dm_bound_app}
\end{equation}

It remains to control the cross term in
\eqref{eq:error_square_recursion_app}. Let
\[
    \xi_m := F_m^\dagger d_m.
\]
Multiplying \eqref{eq:dm_vocf} by
$F_m^\dagger = e^{i\tau A_m}e^{i\tau B_m}$ gives the exact identity
\begin{equation}
    \xi_m
    =
    -i\int_0^\tau
    e^{i\tau A_m}
    e^{isB_m}
    (B_m - A_m)
    e^{-i(s+\tau)A_m}
    |\psi_{m-1}^c\rangle\,ds.
    \label{eq:xim_vocf}
\end{equation}
Let $\mathcal{F}_{m-1}$ be the sigma-algebra generated by the sampled
indices $\{j_{n,1}, j_{n,2}, \dots, j_{n,2(m-1)}\}$ in the first $m-1$
coupled blocks. Conditional on $\mathcal{F}_{m-1}$, the vectors
$e_{m-1}$ and $|\psi_{m-1}^c\rangle$ are fixed, while $A_m =
H_{j_{n,2m-1}}$ and $B_m = H_{j_{n,2m}}$ are independent and
identically distributed.

For $0 \le s \le \tau$, write the integrand operator in
\eqref{eq:xim_vocf} as
\[
    Q_m(s) :=
    e^{i\tau A_m}
    e^{isB_m}
    (B_m - A_m)
    e^{-i(s+\tau)A_m}.
\]
By the integral identity
\[
    e^{iuH} - I = i\int_0^u H e^{irH}\,dr,
    \qquad
    \|e^{iuH} - I\| \le |u|\,\|H\|,
\]
and using $\|A_m\|, \|B_m\| \le 1$, we have
\begin{align}
    \|Q_m(s) - (B_m - A_m)\|
    &\le
    \|e^{i\tau A_m} - I\|\,\|B_m - A_m\|       \notag\\
    &\quad+
    \|e^{isB_m} - I\|\,\|B_m - A_m\|            \notag\\
    &\quad+
    \|B_m - A_m\|\,\|e^{-i(s+\tau)A_m} - I\|    \notag\\
    &\le
    2\tau + 2s + 2(s+\tau)
    \le
    8\tau.
    \label{eq:integrand_remainder_bound_app}
\end{align}
Since $A_m = H_{j_{n,2m-1}}$ and $B_m = H_{j_{n,2m}}$ are i.i.d.,
\[
    \mathbb{E}[B_m - A_m] = 0.
\]
Therefore, taking conditional expectation in \eqref{eq:xim_vocf} and
using \eqref{eq:integrand_remainder_bound_app},
\begin{equation}
    \left\|
    \mathbb{E}\!\left[
        \xi_m\,\middle|\,\mathcal{F}_{m-1}
    \right]
    \right\|
    \le
    8\tau^2.
    \label{eq:xi_conditional_mean_app}
\end{equation}

Taking conditional expectation in
\eqref{eq:error_square_recursion_app}, and using
\eqref{eq:dm_bound_app} and \eqref{eq:xi_conditional_mean_app}, gives
\begin{align}
    \mathbb{E}\!\left[
        \|e_m\|^2\,\middle|\,\mathcal{F}_{m-1}
    \right]
    &\le
    \|e_{m-1}\|^2
    + 2\|e_{m-1}\|
    \left\|
    \mathbb{E}\!\left[
        \xi_m\,\middle|\,\mathcal{F}_{m-1}
    \right]
    \right\|
    + 4\tau^2.
\end{align}
Since $e_{m-1}$ is the difference of two normalised states,
$\|e_{m-1}\| \le 2$. Hence
\begin{equation}
    \mathbb{E}\!\left[
        \|e_m\|^2\,\middle|\,\mathcal{F}_{m-1}
    \right]
    \le
    \|e_{m-1}\|^2 + 36\tau^2.
\end{equation}
Taking full expectation and iterating from $e_0 = 0$ yields
\begin{equation}
    \mathbb{E}\|e_K\|^2
    \le
    36K\tau^2.
    \label{eq:state_mean_square_bound_app}
\end{equation}
Using $K = N_{\ell-1} = N_\ell/2$ and $\tau = \tau_\ell = \lambda t/N_\ell$,
we obtain
\begin{equation}
    \mathbb{E}\|e_K\|^2
    \le
    18\,\frac{\lambda^2t^2}{N_\ell}.
    \label{eq:state_error_final_app}
\end{equation}

Combining \eqref{eq:observable_to_state_error} with
\eqref{eq:state_error_final_app}, we conclude that
\[
    V_\ell
    \le
    \mathbb{E}|P_\ell^{(\mathbf{j}_n,\ell)}
    - P_{\ell-1}^{(\tilde{\mathbf{j}}_n,\ell)}|^2
    \le
    4\,\mathbb{E}\|e_K\|^2
    \le
    72\,\frac{\lambda^2t^2}{N_\ell}.
\]
Finally, since $N_\ell = N_0\,2^\ell$,
\[
    V_\ell
    \le
    72\,\frac{\lambda^2t^2}{N_0}\,2^{-\ell}.
\]
Thus the index-sharing correction variance has decay rate $\beta = 1$.
Equivalently, the bound~\eqref{eq:variance_bound_app} holds with
$C_{\mathrm{IS}} = 72$.
\end{proof}

\section{Complexity Analysis of MLMC-qDRIFT (Proof of Theorem)} \label{app:MLMC_Complexity}

\textbf{Theorem \ref{thm:index_sharing_complexity}.} [Complexity of Index-Sharing ML-qDRIFT] 
Let
\[
    P = \operatorname{Tr}\!\left(O\,e^{-iHt}\rho\,e^{iHt}\right),
\]
and assume $\|H_j\| \le 1$ and $\|O\| \le 1$. Choose the finest level
$L$ so that the qDRIFT bias satisfies
\begin{equation}
    \bigl|\mathbb{E}[P_L] - P\bigr|
    \le \frac{\varepsilon}{\sqrt{2}},
\end{equation}
and choose the sample counts $n_\ell$ according
to~\eqref{eq:optimal_allocation}, so that
\begin{equation}
    \sum_{\ell=0}^{L}\frac{V_\ell}{n_\ell}
    \le \frac{\varepsilon^2}{2}.
\end{equation}
Then the MLMC estimator $\widehat{Y}$ of~\eqref{eq:mlmc_estimator}
satisfies
\begin{equation}
    \mathbb{E}\!\left[(\widehat{Y} - P)^2\right]^{1/2}
    \le \varepsilon.
\end{equation}
Moreover, under the index-sharing variance bound~\eqref{eq:variance_bound},
the expected total gate count obeys
\begin{equation}
    \mathbb{E}[C]
    =
    \mathcal{O}\!\left(
        t^2\lambda^2\,
        \varepsilon^{-2}
        \log^2(1/\varepsilon)
    \right).
\end{equation}

\begin{proof}
We verify the three parameters $\alpha, \beta, \gamma$ of the MLMC
Complexity Theorem~\cite[Theorem~1]{giles2015multilevel} and then
compute the resulting gate count explicitly.

\textbf{Variance decay ($\beta = 1$).}
Lemma~\ref{lemma:index_sharing} gives 
\[
    V_\ell
    \le
    C_{\rm IS}\frac{\lambda^2t^2}{N_0}2^{-\ell},
\]
so $\beta=1$.

\textbf{Cost per sample ($\gamma = 1$).}
Each level-$\ell$ correction requires both the fine circuit of depth
$N_\ell$ and the coarse circuit of depth $N_{\ell-1}$, giving combined
gate count
\begin{equation}
    \mathcal{C}_\ell
    = N_\ell + N_{\ell-1}
    = N_0\,2^\ell + N_0\,2^{\ell-1}
    = \tfrac{3}{2}N_0\,2^\ell,
\end{equation}
which grows as $2^\ell$, so $\gamma = 1$.

\textbf{Bias decay ($\alpha = 1$).}
The Campbell diamond-norm bound~\eqref{eq:qdrift_bias} gives
$|\mathbb{E}[P_L] - P| \le 2\lambda^2 t^2/N_L = \mathcal{O}(2^{-L})$, so
$\alpha = 1$. The condition $\alpha \ge \tfrac{1}{2}\min(\beta,\gamma)
= \tfrac{1}{2}$ is satisfied.

\textbf{Giles sum and total cost.}
Since $\beta = \gamma = 1$, the upper bound on
$\sqrt{V_\ell\,\mathcal{C}_\ell}$ is constant across all levels:
substituting the bounds above,
\begin{equation}
    \sqrt{V_\ell\mathcal C_\ell}
    \le
    \sqrt{
        C_{\rm IS}\frac{\lambda^2t^2}{N_0}2^{-\ell}
        \cdot
        \frac32 N_0 2^\ell
    }
    =
    \sqrt{\frac{3C_{\rm IS}}2}\,\lambda t .
\end{equation}
where $N_0$, $2^{-\ell}$, and $2^\ell$ cancel exactly, leaving a
bound independent of both $\ell$ and $N_0$. The Cauchy--Schwarz optimal allocation~\cite{giles2015multilevel}
then gives
\begin{equation}
    \mathbb{E}[C]
    \le \frac{2}{\varepsilon^2}
        \Bigl(\sum_{\ell=0}^{L}\sqrt{V_\ell\,\mathcal{C}_\ell}\Bigr)^2
    \le \frac{2}{\varepsilon^2}\,(L+1)^2
        \cdot \frac{3C_{\rm IS}}{2}\,t^2\lambda^2
    =
    \frac{3C_{\rm IS}\,t^2\lambda^2\,(L+1)^2}{\varepsilon^2}.
\end{equation}

To suppress the bias to $\varepsilon/\sqrt{2}$, the condition
$2\lambda^2 t^2/N_L \le \varepsilon/\sqrt{2}$ requires
$L = \lceil\log_2(\varepsilon^{-1})\rceil + \mathcal{O}(1)$, so that
$L + 1 \le \tfrac{3}{2}\log_2\varepsilon^{-1}$ for all sufficiently
small $\varepsilon$. Hence $(L+1)^2 \le \tfrac{8}{3}
(\log_2\varepsilon^{-1})^2$, and substituting,
\begin{equation}
    \mathbb{E}[C]
    \le
    \frac{
        3C_{\rm IS}\cdot \tfrac{8}{3}\,
        t^2\lambda^2
    }{\varepsilon^2}
    (\log_2\varepsilon^{-1})^2
    =
    \frac{
        8C_{\rm IS}\,t^2\lambda^2
    }{\varepsilon^2}
    (\log_2\varepsilon^{-1})^2 .
\end{equation}
Thus the total expected gate count is
\[
    \mathbb{E}[C]
    =
    \mathcal{O}\!\left(
        t^2\lambda^2\varepsilon^{-2}
        [\log(1/\varepsilon)]^2
    \right).
\]
The cancellation of $N_0$ in $\sqrt{V_\ell\,\mathcal{C}_\ell}$ follows
from the variance bound $V_\ell=\mathcal{O}(\lambda^2t^2/N_\ell)$ and
the cost growth $\mathcal{C}_\ell=\mathcal{O}(N_\ell)$. Hence the
complexity constant depends on the physical parameters only through
$\lambda t$, up to the universal constant $C_{\rm IS}$.
\end{proof}

\section{Detailed Complexity Derivations}
\label{app:complexity}

\subsection{MLMC-qDRIFT gate count}

The finest level $L$ is chosen to satisfy the bias requirement
$|\mathbb{E}[P_L] - P| \le \varepsilon/\sqrt{2}$.  Since
$N_L = N_0\,2^L$, it is sufficient to take
\begin{equation}
    L
    =
    \left\lceil
    \log_2\!\left(
        \frac{\sqrt{2}B}{\varepsilon N_0}
    \right)
    \right\rceil_+,
    \qquad
    \lceil x \rceil_+ := \max\{0, \lceil x \rceil\},
    \label{eq:L_choice}
\end{equation}
so $L = \mathcal{O}(\log(1/\varepsilon))$.

The cost of one level-$\ell$ sample is
\begin{equation}
    C_0 = N_0,
    \qquad
    C_\ell = N_\ell + N_{\ell-1} = \tfrac{3}{2}N_\ell,
    \quad \ell \ge 1,
    \label{eq:level_cost}
\end{equation}
where $C_\ell$ counts the fine and coupled coarse qDRIFT paths needed
to form one correction sample.  The key cancellation is
\begin{equation}
    \sqrt{V_\ell C_\ell}
    \le
    \sqrt{\frac{A}{N_\ell} \cdot \frac{3}{2}N_\ell}
    =
    \sqrt{\frac{3A}{2}},
    \qquad \ell \ge 1,
    \label{eq:key_cancellation}
\end{equation}
so variance and cost grow at exactly matching rates ($\beta = \gamma = 1$).

The optimal MLMC sample allocation is
\begin{equation}
    n_\ell
    =
    \left\lceil
    \frac{2}{\varepsilon^2}
    \sqrt{\frac{V_\ell}{C_\ell}}\,\mathcal{S}_L
    \right\rceil,
    \qquad
    \mathcal{S}_L
    :=
    \sum_{\ell'=0}^{L} \sqrt{V_{\ell'} C_{\ell'}},
    \label{eq:allocation}
\end{equation}
which gives
\begin{equation}
    \mathcal{S}_L
    \le
    \sqrt{V_0 N_0}
    +
    L\sqrt{\frac{3A}{2}}.
    \label{eq:giles_sum}
\end{equation}
Since $L = \mathcal{O}(\log(1/\varepsilon))$, the total MLMC gate count is
\begin{equation}
    C_{\mathrm{MLMC}}(\varepsilon)
    =
    \frac{2\mathcal{S}_L^2}{\varepsilon^2}
    =
    \mathcal{O}\!\left(\varepsilon^{-2}\log^2(1/\varepsilon)\right),
\end{equation}
as stated in \cref{eq:mlmc_cost}.

\subsection{Crossover regime analysis}
\label{app:crossover}

The crossover precision $\varepsilon^*$ is determined by equating the
leading costs in \cref{eq:std_cost,eq:mlmc_cost}, giving
\begin{equation}
    \frac{\sqrt{2}\,B\sigma^2}{\varepsilon}
    =
    \mathcal{S}_L^2,
    \label{eq:crossover_condition}
\end{equation}
which is generally transcendental due to the $\varepsilon$-dependence of
$\mathcal{S}_L$ through $L$.  Two limiting regimes are tractable.

\textit{Overhead-dominated regime.}
If $\sqrt{V_0 N_0} \gg L\sqrt{3A/2}$, then
$\mathcal{S}_L \approx \sqrt{V_0 N_0}$ and
\begin{equation}
    \varepsilon^*_{\mathrm{oh}}
    \approx
    \frac{\sqrt{2}\,B\sigma^2}{V_0 N_0}.
    \label{eq:crossover_overhead}
\end{equation}
A large base depth $N_0$, a large coarsest-level variance $V_0$, or a
small standard-qDRIFT variance $\sigma^2$ all push the crossover to
smaller $\varepsilon$.

\textit{Correction-level-dominated regime.}
If $L\sqrt{3A/2} \gg \sqrt{V_0 N_0}$, then
$\mathcal{S}_L \approx L\sqrt{3A/2}$ and the crossover condition becomes
\begin{equation}
    \frac{\sqrt{2}\,B\sigma^2}{\varepsilon}
    \approx
    \frac{3A}{2}
    \log_2^2\!\left(
        \frac{\sqrt{2}B}{\varepsilon N_0}
    \right).
    \label{eq:crossover_log}
\end{equation}
Setting $x := \log_2(\sqrt{2}B / (\varepsilon N_0))$ so that
$\varepsilon = \sqrt{2}B/(N_0\,2^x)$, this reduces to the scalar equation
\begin{equation}
    \sigma^2 N_0\,2^x
    \approx
    \frac{3A}{2}\,x^2,
    \label{eq:crossover_x}
\end{equation}
which is solved numerically in \cref{sec:numerics}.  For systems with
large effective one-norm $\lambda t$, the constants $B$, $A$, and $N_0$
are all large, and the overhead-dominated regime can persist to relatively
small $\varepsilon$.  Composite deterministic--randomized
decompositions~\cite{jin2025partially, hagan2023composite} reduce the
effective one-norm seen by the randomized layer, lowering these constants
and shifting $\varepsilon^*$ toward more practical precision targets.

\end{document}